\documentclass[11pt]{amsart}
\usepackage{amsmath,amssymb,amsthm,mathscinet}
\usepackage{enumitem}
\usepackage{cite,hyperref}
\usepackage[noabbrev,capitalize]{cleveref}

\allowdisplaybreaks

\renewcommand{\epsilon}{\varepsilon}
\renewcommand{\phi}{\varphi}
\DeclareMathOperator{\alg}{alg}
\DeclareMathOperator{\GF}{GF}
\newcommand{\C}{{\mathbb C}}
\newcommand{\F}{{\mathbb F}}
\newcommand{\N}{{\mathbb N}}
\newcommand{\R}{{\mathbb R}}
\newcommand{\Z}{{\mathbb Z}}
\newcommand{\Fthree}{\F_3}
\newcommand{\Fp}{\F_p}
\newcommand{\Fpu}{\F_p^*}
\newcommand{\Fu}{F^*}
\newcommand{\Fqr}{F^{*2}}
\newcommand{\Eu}{E^*}
\newcommand{\acF}{F^{\alg}}
\newcommand{\acFu}{F^{\alg*}}
\newcommand{\du}{\Fu \sqcup U_E}
\newcommand{\dc}{\lambda}
\newcommand{\card}[1]{\left|{#1}\right|}
\newcommand{\cardF}{|F|}
\newcommand{\cardFu}{|\Fu|}
\newcommand{\sums}[1]{\sum_{\substack{#1}}}
\newtheorem{theorem}{Theorem}[section]
\newtheorem{proposition}[theorem]{Proposition}
\newtheorem{lemma}[theorem]{Lemma}
\newtheorem{corollary}[theorem]{Corollary}
\theoremstyle{definition}
\newtheorem{definition}[theorem]{Definition}
\newtheorem{example}[theorem]{Example}
\newtheorem{remark}[theorem]{Remark}
\newtheorem{step}{Step}
\newtheorem{gradus}{Step}

\begin{document}

\title[APN power functions with exponents expressed as fractions]{Almost perfect nonlinear power functions with exponents expressed as fractions}
\author[Katz, O'Connor, Pacheco, and Sapozhnikov]{Daniel J.~Katz, Kathleen R.~O'Connor, Kyle Pacheco, and Yakov Sapozhnikov}\thanks{Daniel J.~Katz is and Kathleen R.~O'Connor, Kyle Pacheco, and Yakov Sapozhnikov were with the Department of Mathematics, California State University, Northridge.  Yakov Sapozhnikov is with Schweitzer Engineering Laboratories, Inc.  This paper is based upon work supported in part by the National Science Foundation under Grants 1500856, 1815487, and 2206454.}
\date{04 June 2026}

\begin{abstract}
Let $F$ be a finite field, let $f$ be a function from $F$ to $F$, and let $a$ be a nonzero element of $F$.
The discrete derivative of $f$ in direction $a$ is $\Delta_a f \colon F \to F$ with $(\Delta_a f)(x)=f(x+a)-f(x)$.
The differential spectrum of $f$ is the multiset of cardinalities of all the fibers of all the derivatives $\Delta_a f$ as $a$ runs through $F^*$.
An almost perfect nonlinear (APN) function is one for which the largest cardinality in its differential spectrum is $2$.
Almost perfect nonlinear functions are of interest as cryptographic primitives.
If $d$ is a positive integer, then the power function over $F$ with exponent $d$ is the function $f \colon F \to F$ with $f(x)=x^d$ for every $x \in F$.
There is a small number of known infinite families of APN power functions.
In this paper, we re-express the exponents for one such family in a more convenient form.
This enables us not only to obtain the differential spectrum of each power function $f$ with an exponent in our family, but also to determine the elements that lie in an arbitrary fiber of the discrete derivative of $f$.
This differential analysis, which is far more detailed than previous results, is achieved by composing the discrete derivative of $f$ with some permutations and a double covering of its domain to obtain a function whose fibers can more readily be analyzed.
\end{abstract}

\maketitle

\section{Introduction}

If $F$ is a finite field, then a {\it power function over $F$} is a function $f\colon F \to F$ with $f(x)=x^d$ where $d$ is a positive integer.
Power functions are useful as cryptographic primitives, as they can be used to introduce nonlinearity into systems and can be evaluated quickly.
For example, the sole source of nonlinearity in the Advanced Encryption Standard (AES) is the power function $x\mapsto x^{254}$ over the finite field with $256$ elements.
This function is employed in the portion of the algorithm known as the substitution box (S-box).
Cryptosystems are subjected to attacks that exploit statistically predictable behaviors of the functions they employ.
Differential cryptanalysis focuses on the relation between the difference of inputs and the difference of outputs for a function, with the most resilient functions being those whose behavior is furthest from linear functions, for which the relation between difference of inputs and difference of outputs is invariable and transparent.
This paper concerns a class of power functions whose differential properties are very far from linear.

This paper is especially concerned with power functions whose exponents are expressed as ratios, and to this end, we set a convention about how these ratios are interpreted.
First of all, it should be noted that if $d$ and $e$ are positive integers, and $x\mapsto x^d$ and $x\mapsto x^e$ are functions from a finite field $F$ to itself, then these functions equal each other if and only if $d \equiv e \pmod{\cardFu}$, so although exponents should be considered positive integers, it suffices to specify them up to congruence modulo $\cardFu$.
A positive rational number $r$ can be understood as an exponent for a power function over the finite field $F$ if and only if $r$ can be written as a reduced fraction $d_1/d_2$ of integers $d_1$ and $d_2$ with $\gcd(d_2,\cardFu)=1$, in which case we interpret $x \mapsto x^r$ over $F$ as $x \mapsto x^d$ where $d$ is a positive integer with $d \equiv d_1 d_2^{-1} \pmod{\cardFu}$.
If we remove the condition that $d_1/d_2$ must be reduced in the previous sentence, then the value of $d$ modulo $\cardFu$ (and hence the same power function $x\mapsto x^d$ over $F$) is independent of our choice of $d_1$ and $d_2$ (but we must maintain the condition that $\gcd(d_2,\cardFu)=1$).
For example, if $F$ is the finite field of order $3^3$, then the exponent $d=(3^3+1)/(3^2+1)=28/10$ is reduced to $14/5$, which is interpreted as $14 \cdot 5^{-1} \pmod{3^3-1}$, that is, $14 \cdot 21 \equiv 8 \pmod{26}$, so the power function over $F$ with exponent $d$ is $f(x)=x^8$.

If $f\colon A \to B$ is any function and $b \in B$, then the {\it fiber of $f$ over $b$} is the set $f^{-1}(\{b\})=\{a \in A: f(a)=b\}$.
If $F$ is a field and $f\colon F \to F$ and $a \in \Fu$, then $\Delta_a f$ is the function from $F$ to $F$ with $(\Delta_a f)(x)=f(x+a)-f(x)$ for every $x \in F$, and we call $\Delta_a f$ the {\it discrete derivative of $f$ in direction $a$}.
We define $\Delta=\Delta_1$, and if we omit mention of a direction, then the {\it discrete derivative of $f$} is just $\Delta f$ with $(\Delta f)(x)=f(x+1)-f(x)$.

If $F$ is a finite field, then for $(a,b) \in \Fu\times F$ the {\it differential multiplicity of $f$ for $(a,b)$}, written $\delta_f(a,b)$, is the cardinality of the fiber of $\Delta_a f$ over $b$, i.e., $\delta_f(a,b)=\card{(\Delta_a f)^{-1}(\{b\})}=\card{\{x \in F: f(x+a)-f(x)=b\}}$.
It is equal to the number of solutions $(x,y) \in F\times F$ to the system of equations
\begin{align*}
y-x & = a\\
f(y)-f(x) & = b,
\end{align*}
and, as such, tells us how many ways the difference of outputs can equal $b$ when the difference of inputs equals $a$.
The {\it differential spectrum of $f$} is the multiset of all $\delta_f(a,b)$ as $(a,b)$ runs through $\Fu\times F$, and {\it differential uniformity of $f$} is the largest differential multiplicity.
Low differential uniformity makes a function resistant to differential cryptanalysis \cite{Nyberg} because it indicates that no difference of outputs is strongly correlated to any difference of inputs.
There is a great deal of interest in functions of differential uniformity $1$, called perfect nonlinear (PN) functions, and functions of differential uniformity $2$, called almost perfect nonlinear (APN) functions.

If $f\colon F \to F$ is the power function with $f(x)=x^d$ and $a \in \Fu$, then $(\Delta_a f)^{-1}(\{b\})$ can be obtained from $(\Delta f)^{-1}(\{b/a^d\})$ by scaling all the elements of the latter set by $a$, and thus $\delta_f(a,b)=\delta_f(1,b/a^d)$.
Therefore, we define the {\it reduced differential spectrum of $f$} as the multiset of all $\delta_f(1,c)$ as $c$ runs through $F$: it has the same elements as the full differential spectrum, but the number of instances of each value in the reduced spectrum is smaller by a factor of $\cardFu$.
Simply put, the reduced differential spectrum of $f$ is the multiset of cardinalities of the fibers of $\Delta f$.

If $t$ and $n_1,\ldots,n_t$ are nonnegative integers and $a_1,\ldots,a_t$ are distinct elements, then $n_1 [a_1]+\cdots+n_t [a_t]$ represents the multiset of elements drawn from $\{a_1,\ldots,a_t\}$ with precisely $n_j$ instances of element $a_j$ for each $j \in \{1,\ldots,t\}$.
Note that the reduced differential spectrum of a function $f\colon F \to F$ always contains $\cardF$ nonnegative integers whose sum is $\cardF$ since the fibers of a function form a disjoint cover of the domain, and so the average fiber size is $1$.
In particular, if $f$ is an APN power function, then the fiber sizes can only be $0$, $1$, and $2$, so the reduced differential spectrum must be of the form
\[
\frac{\cardF-N}{2} \, [0] + N \, [1] + \frac{\cardF-N}{2} \, [2]
\]
for some integer $N$ with $0\leq N < \cardF$.
If $f(x)=x^d$ and $\cardF$ is even or $d$ is odd, then it is easy to check that every fiber of $\Delta f$ is closed under the involution $x\mapsto -1-x$, which has no fixed point if $\cardF$ is even and only one fixed point, $-1/2$, if $\cardF$ is odd.
This makes $N=0$ when $\cardF$ is even and $N=1$ when $\cardF$ and $d$ are both odd.
The case when $\cardF$ is odd and $d$ is even is more difficult: $N$ must be an odd number (since $(\cardF-N)/2$ is an integer), but $N$ can be larger than $1$, and there are situations where the determination of $N$ is nontrivial, e.g., for the functions discussed in the main result of this paper.

The purpose of this paper is to examine the fibers of the discrete derivatives of an infinite family of APN power functions over finite fields of characteristic $3$.
To understand our theorem, one should recall that when $F$ is a finite field of odd characteristic, the extended quadratic character maps quadratic residues of $F$ to $1$, quadratic nonresidues of $F$ to $-1$, and the zero element of $F$ to $0$.
We now state our main result.
\begin{theorem}\label{Victoria}
Let $F$ be a finite field of order $q=3^n$ with $n$ odd and let $k$ be an even nonnegative integer with $\gcd(k,n)=1$.
Let $f\colon F \to F$ be the power function with exponent $(3^n+1)/(3^k+1)$.
If $\eta$ is the extended quadratic character of $F$, then we have the following determination of the fibers of $\Delta f$:
\begin{enumerate}[label=(\roman*)]
\item If $c \in F$ with $\eta(1-c^{3^k+1})=-1$, then $(\Delta f)^{-1}(\{c\})=\emptyset$.
\item If $c \in F$ with $\eta(1-c^{3^k+1})\not=-1$, then $(\Delta f)^{-1}(\{c\})=\{p_{10}(c),p_{11}(c)\}$, where $p_{10}$ and $p_{11}$ are polynomials in $\Fthree[x]$ defined just after this theorem.
If $c\in\Fthree$, then $p_{10}(c)=p_{11}(c)=1-c$, but if $c\not\in\Fthree$, then $p_{10}(c)$ and $p_{11}(c)$ are distinct elements of $F\smallsetminus\Fthree$.
\end{enumerate}
Therefore, for any $c \in \Fu$, we have $\card{(\Delta f)^{-1}(\{c\})}=1+\eta(1-c^{3^k+1})$, and $\card{(\Delta f)^{-1}(\{0\})}=1$, and the reduced differential spectrum of $f$ is
\[
\frac{3^n-3}{2} [0] + 3 [1] + \frac{3^n-3}{2} [2],
\]
which makes $f$ APN if $n>1$ (but $f$ is PN if $n=1$).
\end{theorem}
The polynomials $p_{10}$ and $p_{11}$ in \cref{Victoria} are defined by letting $\epsilon$ be a positive integer that is a multiplicative inverse of $(3^k-1)/4$ modulo $(q-1)/2$ (with $k$ and $q$ as in the theorem), and setting
\begin{align*}
p_1(x) & = (1-x^{3^k+1})^{\frac{q+1}{4}}, \\
p_2(x) & = x^{q-2+\frac{(q-3)\epsilon}{4}} \left(p_1(x)+1 \right)^{\frac{(3 q-1)\epsilon}{4}}, \\
p_3(x) & = p_2(x)+p_2(x)^{q-2}, \\
p_4(x) & = \prod_{j=0}^{n-1} \left((-1)^j p_1(x)^{3^{j k}}+1\right), \\
p_5(x) & = \left(p_4(x)-p_4(x)^{q-2}\right) \left(p_4(x)+p_4(x)^{q-2}\right)^{q-2}, \\
p_6(x) & = \left(1-p_5(x)^2\right)^{\frac{q+1}{4}}, \\
p_7(x) & = \left((x-1)^{\frac{q+1}{2}}-(x+1)^{\frac{q+1}{2}}\right)^{\frac{3^k+1}{2}} \left((x-1)^{\frac{3^k+1}{2}}-(x+1)^{\frac{3^k+1}{2}}\right)^{\frac{q-3}{2}}, \\
p_8(x) & = x^{\frac{2 q-3^k-3}{2}} p_7(p_6(x)) p_6(x), \\
p_9(x) & = (((x+1)^{q-2}+1)^{\frac{3^k+1}{2}}-1)^{q-2}, \\
p_{10}(x) & = p_9(p_3(x))\text{, and} \\
p_{11}(x) & = p_9(p_8(x)).
\end{align*}

The APN power functions described by this theorem were discussed under a different guise in other works.
We say that two exponents $d$ and $e$ are {\it equivalent} over a finite field $F$ of characteristic $p$ to mean that there is some $j \in \Z$ and $k \in \{1,-1\}$ such that $d\equiv p^j e^k \pmod{\cardFu}$, where we are only allowed to use $k=-1$ if $\gcd(e,\cardFu)=1$.  This gives an equivalence relation on positive integers.  For a long time it has been known (see, for example, \cite[pp.~475, 476]{Helleseth-Rong-Sandberg} and \cite[p.~1271]{Dobbertin}) that when exponents are equivalent they produce the same differential spectrum. 
Power functions with exponents equivalent to those in \cref{Victoria} are studied in the work of Zha and Wang \cite[Theorem 4.1]{Zha-Wang}, which indicates that they are APN (see \cref{Reginald} for details) without determining the differential spectrum.
The later work of Tian and Chen \cite{Tian-Chen} appears to give an independent proof that the functions are APN and to determine the size of each fiber.
Our \cref{Victoria} goes much further in actually furnishing polynomial formulae to determine the elements of the fibers.
These formulae require $O(\log(\cardF))$ multiplications, additions, and exponentiations in the finite field to obtain the exact contents of a fiber, while a direct calculation would require $O(\cardF)$ such operations.

The rest of this paper is organized as follows.
\cref{Eugene} provides preliminaries, i.e., definitions, notations, conventions, and basic general results.
Then we proceed to the proof of our result, \cref{Victoria}, which relies on the following main ideas:
\begin{enumerate}
\item \cref{Fabian} describes how composition with permutations transforms fibers of maps, and then uses specific permutations to transform our discrete derivative into a more tractable function.
\item This is not sufficient to achieve the analysis in \cref{Victoria}, so we need to go further afield and compose with a double covering, which is described in \cref{David}.
\item In \cref{Helen} we show how to determine precisely the elements in each fiber of the composition of the discrete derivative of our power function with the permutations from \cref{Fabian} and the double covering from \cref{David}.
\item We then use this vantage point in \cref{Prospero} to obtain the complete determination of the fibers of the discrete derivative, i.e., we prove \cref{Victoria}.
\end{enumerate}
We close with an appendix, which shows the equivalence between the exponents studied in \cref{Victoria} and those studied by Zha and Wang in \cite[Theorem 4.1]{Zha-Wang}.

\section{Preliminaries}\label{Eugene}

In this section, we give basic definitions, notations, and conventions used throughout the paper, as well as some basic general results that are useful later.
We use $\N$ to denote the set of nonnegative integers.
If $F$ is a field, then we use $\Fu$ to denote the multiplicative group of units of $F$ and $\acF$ to denote an algebraic closure of $F$.
We break the rest of this section into smaller sections on $p$-adic valuations (\ref{Valerie}), greatest common divisors (\ref{Greta}), quadratic residues and nonresidues (\ref{Robert}), power functions and roots (\ref{Ferdinand}), and some algebraic results (\ref{Algernon}) used in our proofs.

\subsection{Valuations}\label{Valerie}

In this section, we define the notion of a $p$-adic valuation and then analyze the $2$-adic valuations of terms that commonly occur in this work.
\begin{definition}
Let $n \in \Z$ and let $p$ be a prime number. If $n\not=0$, then the \textit{$p$-adic valuation of $n$}, denoted $v_p(n)$, is the unique $k \in \N$ such that $p^k|n$ and $p^{k+1}\nmid n$.
If $n=0$, then we decree $v_p(0)=\infty$, where $\infty$ is treated as $+\infty$ from the extended real number system.
\end{definition}
Because we are analyzing power functions with exponents of the form $(p^j+1)/(p^k+1)$ where $p$ is an odd prime, it is often necessary to understand the $2$-adic valuation of terms of the form $p^j+1$ and $p^j-1$, which are well-known cases of what is sometimes called the lifting-the-exponent lemma. 
\begin{lemma}\label{Francis}
Let $a \in \Z$ and let $j$ be a nonnegative integer.
\begin{enumerate}[label=(\roman*)]
\item If $a\equiv 1 \pmod{4}$, then
\begin{align*}
v_2(a^j+1) & = 1 \text{ and}\\
v_2(a^j-1) & = v_2(a-1)+v_2(j).
\end{align*}
\item
If $a \equiv 3 \pmod{4}$, then
\begin{align*}
v_2(a^j+1) & =
\begin{cases}
1 & \text{if $j$ is even,} \\
v_2(a+1) & \text{if $j$ is odd,}
\end{cases} \\
v_2(a^j-1) & =
\begin{cases}
v_2(a+1)+v_2(j) & \text{if $j$ is even,} \\
1 & \text{if $j$ is odd.}
\end{cases}
\end{align*}
\end{enumerate}
\end{lemma}

\subsection{Greatest common divisors}\label{Greta}

In our analysis of power functions with exponents of the form $(p^j+1)/(p^k+1)$ over the finite field of order $p^n$ where $p$ is an odd prime, it is often necessary to compute greatest common divisors involving terms of the form $p^m+1$ and $p^n-1$ (the latter being the order of the unit group of our field).
The first claim in the following lemma is well known, and although the latter claims are not as well known, they can be deduced in a straightforward manner from the results of \cite{Cade-Lau-Pedersen-Lossers} and \cite{Penney-Stein}.

\begin{lemma}\label{Clementine}
Let $a$ be an integer, and let $m$ and $n$ be nonnegative integers.
\begin{enumerate}[label=(\roman*)]
\item\label{Hortense} Then
\[
\gcd(a^m-1,a^n-1)=|a^{\gcd(m,n)}-1|.
\]
\item Also
\[
\gcd(a^m+1,a^n+1)=
\begin{cases}
|a^{\gcd(m,n)}+1| & \text{if $v_2(m) = v_2(n)$,} \\
2 & \text{if $v_2(m)\not=v_2(n)$ and $a$ is odd,} \\
1 & \text{if $v_2(m)\not=v_2(n)$ and $a$ is even.}
\end{cases}
\]
\item And also
\[
\gcd(a^m+1,a^n-1)=
\begin{cases}
|a^{\gcd(m,n)}+1| & \text{if $v_2(m) < v_2(n)$,} \\
2 & \text{if $v_2(m) \geq v_2(n)$ and $a$ is odd,} \\
2 & \text{if $(m,n)=(0,0)$ and $a$ is even,} \\
1 & \text{otherwise.}
\end{cases}
\]
\item\label{Arthur} Furthermore, if $a$ is odd, then
\[
\gcd\left(\frac{a^m+1}{2},a^n-1\right)=
\begin{cases}
\frac{|a^{\gcd(m,n)}+1|}{2} & \text{if $v_2(m) < v_2(n)$,} \\[2pt]
2 & \text{if $a \equiv 3 \!\!\!\!\! \pmod{4}$ and $m,n$ are odd,} \\
1 & \text{otherwise.}
\end{cases}
\]
\end{enumerate}
\end{lemma}

\subsection{Quadratic residues and nonresidues}\label{Robert}

For a finite field $F$ of odd characteristic, a {\it quadratic residue of $F$} is an element of the form $a^2$ for some $a \in \Fu$ and a {\it quadratic nonresidue of $F$} is an element of $\Fu$ that is not a quadratic residue.
The $0$ element is considered neither a quadratic residue nor a quadratic nonresidue.
We let $\Fqr$ denote the set of quadratic residues of $F$.
There are $\cardFu/2$ quadratic residues and $\cardFu/2$ quadratic nonresidues in $F$.
If $a \in F$, then
\[
a^{\cardFu/2}=
\begin{cases}
0 & \text{if $a=0$,} \\
1 & \text{if $a$ is a quadratic residue,} \\
-1 & \text{if $a$ is a quadratic nonresidue.}
\end{cases}
\]
The {\it extended quadratic character of $F$} is the map $\eta\colon F \to \C$, that maps $0$ to $0$, quadratic residues to $1$, and quadratic nonresidues to $-1$.  Thus, if $F$ is of characteristic $p$ and order $q$, then $a^{(q-1)/2}$ is equal to $\eta(a) \pmod{p}$ for every $a \in F$.
We therefore have $\eta(a b)=\eta(a)\eta(b)$ for every $a,b \in F$.  If  we remove $0$ from the domain and codomain of $\eta$, we obtain the {\it quadratic character of $F$}, which is a homomorphism from the multiplicative group of $F$ onto the multiplicative subgroup $\{1,-1\}$ of $\C$.

It is often important to ascertain whether $-1$ is a quadratic residue, and we have the following well-known criterion.
\begin{lemma}\label{Quentin}
Let $F$ be a finite field of odd characteristic $p$ and order $p^n$.  The element $-1$ is a quadratic nonresidue in $F$ if $\cardF \equiv 3 \pmod{4}$ (which happens if and only if $p\equiv 3 \pmod{4}$ and $n$ is odd); otherwise $-1$ is a quadratic residue.
\end{lemma}
Then the multiplicativity of the quadratic character gives the following.
\begin{corollary}\label{Alasdair}
Let $F$ be a finite field with $\cardF \equiv 3 \pmod{4}$ and let $x \in \Fu$.  Then exactly one element of $\{x,-x\}$ is a quadratic residue and the other is a quadratic nonresidue.
\end{corollary}
When $-1$ is a quadratic nonresidue in a field, it helps to have a polynomial formula for singling out which of a pair of distinct opposite elements is the quadratic residue and which is the quadratic nonresidue.
\begin{lemma}\label{Lily}
Let $F$ be a finite field with $\cardF\equiv 3\pmod{4}$ and let $\eta\colon F \to \C$ be the quadratic character of $F$.  If $x,y \in \Fu$, then the unique element $z \in \{x,-x\}$ with $\eta(z)=\eta(y)$ is $y^{(\cardF-1)/2} x^{(\cardF+1)/2}$.
If $x \in \Fqr$, then $z=y^{(\cardF-1)/2} x$.
\end{lemma}
\begin{proof}
Let $w=y^{(\cardF-1)/2} x^{(\cardF+1)/2}$.
Then $w/x$ is equal to $(x y)^{(\cardF-1)/2}$, which is in $\{1,-1\}$, so $w \in \{x,-x\}$.  Also $w y = (x y)^{(\cardF+1)/2}$ so that $\eta(w y)$ is congruent modulo the characteristic of $F$ to
\[
(w y)^{(\cardF-1)/2}=((x y)^{\cardF-1})^{(\cardF+1)/4}=1^{(\cardF+1)/4}=1,
\]
where we make use of the fact that $4$ divides $\cardF+1$ because of our given assumption about the cardinality of $F$.
Thus, $\eta(w y)=1$, and so $w$ is one element of $\{x,-x\}$ with $\eta(w)=\eta(y)$, but there is no other such element by \cref{Alasdair}, so $z=w$.
If $x\in\Fqr$, then $x^{(\cardF-1)/2}=1$, making $z=y^{(\cardF-1)/2} x$.
\end{proof}
The following result considers certain sets that are defined using the quadratic character; these feature in a later result (\cref{Uriah}) that is a critical step in establishing our main result.
\begin{lemma}\label{James}
Let $F$ be a finite field of odd characteristic and order $q$ with extended quadratic character $\eta$.
Let $E$ be the quadratic extension of $F$, $B=\{x \in F: \eta(x)\not=1, \eta(x+1)\not=-1\}$, and $\Gamma=\{x \in E: x^2 \in B\}$.
Then $-1 \in B$ if and only if $q\equiv 3 \pmod{4}$.
Both primitive fourth roots of unity lie in $\Gamma$ if $q \equiv 3 \pmod{4}$; otherwise, neither of them lies in $\Gamma$.
Furthermore,
\begin{align*}
|B|
& = \begin{cases}
\frac{q+3}{4} & \text{if $q\equiv 1\pmod{4}$,} \\[2pt]
\frac{q+5}{4} & \text{if $q\equiv 3\pmod{4}$,}
\end{cases} \\
|\Gamma|
& = \begin{cases}
\frac{q+1}{2} & \text{if $q\equiv 1\pmod{4}$,} \\[2pt]
\frac{q+3}{2} & \text{if $q\equiv 3\pmod{4}$.}
\end{cases}
\end{align*}
\end{lemma}
\begin{proof}
It is clear that $0 \in B$.
\cref{Quentin} shows that $-1 \in B$ if and only if $q \equiv 3 \pmod{4}$, whence follows the claim about primitive fourth roots and $\Gamma$.
We can count the elements of $F\smallsetminus\{0,-1\}$ in $B$ using the character sum
\begin{align*}
\sum_{x \in F\smallsetminus\{0,-1\}} & \left(\frac{1-\eta(x)}{2}\right) \left(\frac{1+\eta(x+1)}{2}\right) \\
&=\frac{1}{4} \sum_{x \in F\smallsetminus\{0,-1\}} \left(1-\eta(x)+\eta(x+1)-\eta(x(x+1))\right) \\
&=\frac{q-2+\eta(-1)}{4},
\end{align*}
where the last equality uses the fact that there are equally many quadratic residues and nonresidues in $\Fu$ to sum $\eta(x)$ and $\eta(x+1)$, and uses \cite[Thm.~5.48]{Lidl-Niederreiter} to sum $\eta(x(x+1))$.
We use \cref{Quentin} to evaluate $\eta(-1)$ and recall which the elements of $\{0,-1\}$ lie in $B$ to conclude that if $q\equiv 1 \pmod{4}$, there are $1+(q-1)/4$ elements in $B$, and if $q\equiv 3\pmod{4}$, there are $2+(q-3)/4$ elements in $B$.
Since every nonzero element of $F$ has two square roots in $E$, while $0$ only has one square root in $E$, we see that $|\Gamma|=2|B|-1$.
\end{proof}
The following is a congruence used later to determine whether a certain element is a quadratic residue or nonresidue over a prime field.
\begin{lemma}\label{Molly}
Let $p$ be an odd prime and $j$ a nonnegative integer.  Then $(p^j-1)/2 \equiv j(p-1)/2 \pmod{p-1}$.
\end{lemma}
\begin{proof}
What we want to prove is equivalent to $p^j-1 \equiv j (p-1) \pmod{2(p-1)}$.
Notice that $p^j-1 =(p-1) \sum_{k=0}^{j-1} p^k$, which is $p-1$ times a sum of $j$ odd numbers, so it is an even (resp., odd) multiple of $p-1$ if $j$ is even (resp., odd).  The same can be said of $j(p-1)$, so they match modulo $2(p-1)$.
\end{proof}

\subsection{Power functions and roots}\label{Ferdinand}

This section examines the fibers of power functions $x \mapsto x^d$, which is another way of looking at $d$th roots.
We begin with a general result on the structure of the fibers of a power function, which is a well known consequence of the theory of finite cyclic groups, so we record it without proof.
\begin{lemma}\label{Herbert}
Let $F$ be a finite field of order $q$ with primitive element $\alpha$, $d$ a positive integer, and $f\colon F \to F$ with $f(x)=x^d$.  Let $g=\gcd(d,\cardFu)$.  Let $G$ be the unique cyclic subgroup of $\Fu$ of order $g$, which is generated by $\alpha^{(q-1)/g}$.  If $c \in F$, then precisely one of the following cases holds.
\begin{enumerate}[label=(\roman*)]
\item If $c=0$, then $f^{-1}(\{c\})=\{0\}$.
\item If $c$ is the $g$th power of some element of $\Fu$, then $f^{-1}(\{c\})$ is a coset of $G$ in $\Fu$, and in particular, $f^{-1}(\{1\})=G$.
\item Otherwise, $f^{-1}(\{c\})$ is empty.
\end{enumerate}
In particular, $f$ is a permutation if and only if $g=1$.
\end{lemma}
It is helpful to have an explicit polynomial formula for the square root.  This assists us in writing the polynomial formulae in our main result (\cref{Victoria}).
\begin{corollary}\label{Bernard}
Let $F$ be a finite field with $\cardF\equiv 3 \pmod{4}$ and let $x \in F$.
If $x\in\Fqr\cup\{0\}$, then $x^{(\cardF+1)/4}$ is a square root of $x$.
Furthermore, if $x\in\Fqr$, then $x^{(\cardF+1)/4}$ is the unique square root of $x$ that is a quadratic residue in $F$.
\end{corollary}
\begin{proof}
When $x=0$ the claim is trivial, so assume that $x\in\Fqr$ henceforth.
By \cref{Herbert} and the fact that the unique subgroup of order $2$ in $\Fu$ is $\{1,-1\}$, each quadratic residue has precisely two square roots, which are opposites of each other, so by \cref{Alasdair} precisely one of the square roots of $x$ is a quadratic residue.
Notice that $(x^{(\cardF+1)/4})^4 = x^{\cardF+1}=x^2$.
Thus, $(x^{(\cardF+1)/4})^2 \in \{x,-x\}$, but by \cref{Alasdair}, $-x$ is a quadratic nonresidue, so $(x^{(\cardF+1)/4})^2=x$.
Since  $x^{(\cardF+1)/4}$ is a power of a quadratic residue, it must itself be a quadratic residue.
\end{proof}
The following lemma analyzes the fibers of very specific power functions that arise in the proofs leading up to \cref{Victoria}.
For these functions, one has polynomial formulae for the elements of the fibers.
\begin{lemma}\label{Yolanda}
Let $F$ be a finite field of order $q=3^n$ with $n$ odd, let $k$ be an even nonnegative integer with $\gcd(k,n)=1$, and let $e_2=(3^k-1)/2$.
Then $\gcd(e_2,\cardFu)=2$.
Let $\epsilon$ be a positive integer that is a multiplicative inverse of $e_2/2$ modulo $\cardFu/2$.
Let $g \colon F \to F$ be the power map $g(x)=x^{e_2}$.
For each $c \in F$, we have
\[
g^{-1}(\{c\})=
\begin{cases}
\{0\} & \text{if $c=0$,} \\
\{c^{\epsilon(q+1)/4},-c^{\epsilon(q+1)/4}\}   & \text{if $c\in\Fqr$,}\\
\emptyset & \text{otherwise,}
\end{cases}
\]
and in the second case, the element $c^{\epsilon(q+1)/4}$ is a quadratic residue while $-c^{\epsilon(q+1)/4}$ is a quadratic nonresidue.
\end{lemma}
\begin{proof}
We have $\gcd(e_2,\cardFu) \mid \gcd(3^k-1,3^n-1) = 3^{\gcd(k,n)}-1=2$ by \cref{Clementine}\ref{Hortense}.
Both $e_2$ and $\cardFu$ are even by \cref{Francis}, so $\gcd(e_2,\cardFu)=2$ and $\gcd(e_2/2,\cardFu/2)=1$, thus ensuring that $\epsilon$ is defined.
\cref{Herbert} proves our claims about $g^{-1}(\{c\})$ when $c$ is $0$ or a quadratic nonresidue, and it proves that $g^{-1}(\{c\})$ is a coset of $\{1,-1\}$ in $\Fu$ when $c$ is a quadratic residue, which we assume to be the case for the rest of this proof.
From our definition of $\epsilon$ we have $\epsilon e_2 \equiv 2 \pmod{\cardFu}$, and so $g(c^{\epsilon(q+1)/4})=c^{\epsilon e_2 (q+1)/4}=(c^{(q+1)/4})^2=c$ by \cref{Bernard}, and so we conclude that $g^{-1}(\{c\})=\{c^{\epsilon(q+1)/4},-c^{\epsilon(q+1)/4}\}$.
\cref{Bernard} also shows that $c^{(q+1)/4}\in\Fqr$, and so $c^{\epsilon(q+1)/4}\in\Fqr$, while $-c^{\epsilon(q+1)/4}$ is a quadratic nonresidue by \cref{Alasdair}.
\end{proof}

\subsection{Algebraic lemmas}\label{Algernon}

Two algebraic results that will be useful later are proved here.
\cref{Oscar} is used to prove \cref{Rita}, which is used later in \cref{Beverly}, where we analyze the fibers of discrete derivatives of power functions with exponents of the form $(p^j+1)/(p^k+1)$ where $p$ is a prime number.
\begin{lemma}\label{Oscar}
Let $p$ be a prime and let $j$ be a nonnegative integer.
Then the polynomial $(x+1)^{p^j+1}-(x-1)^{p^j+1} \in \F_p[x]$ is equal to $2(x^{p^j}+x)$.
\end{lemma}
\begin{proof}
We have
\begin{align*}
(x+1)^{p^j+1}-(x-1)^{p^j+1}
& = (x+1)^{p^j}(x+1)-(x-1)^{p^j}(x-1) \\
& = (x^{p^j}+1)(x+1)-(x^{p^j}-1)(x-1) \\
& = 2 (x^{p^j}+x).\qedhere
\end{align*}
\end{proof}
The next lemma uses the idea substituting $x+x^{-1}$ in powers of $x+2$ and $x-2$ and then completing the square to obtain formulae whose exponents are doubled.
This idea was used to great effect by Coulter and Matthews to discover new perfect nonlinear functions in fields of characteristic $3$ (see the proof of their Theorem 4.1 in \cite{Coulter-Matthews}).
Interest in the transformation $\rho(x)=x+x^{-1}$ goes back much further.
For example, the $n$th Dickson polynomial $D_n(x,y)$ has the property that $D_n(\rho(x),1)=\rho(x^n)$ (see \cite[eq.~(7.8)]{Lidl-Niederreiter}).
\begin{lemma}\label{Rita}
Let $p$ be a prime and let $j$ be a nonnegative integer.
Then the Laurent polynomial $(x+x^{-1}+2)^{(p^j+1)/2}-(x+x^{-1}-2)^{(p^j+1)/2} \in \F_p[x,x^{-1}]$ is equal to $2 x^{-(p^j+1)/2} (x^{p^j}+x)$.
\end{lemma}
\begin{proof}
We note that $(x+x^{-1}+2)^{(p^j+1)/2}-(x+x^{-1}-2)^{(p^j+1)/2}$ is equal to 
\begin{align*}
\frac{(x^2+2 x+1)^{(p^j+1)/2}-(x^2-2 x+1)^{(p^j+1)/2}}{x^{(p^j+1)/2}} 
& = \frac{(x+1)^{p^j+1}-(x-1)^{p^j+1}}{x^{(p^j+1)/2}} \\
& = \frac{2(x^{p^j}+x)}{x^{(p^j+1)/2}},
\end{align*}
where the last equality uses \cref{Oscar}.
\end{proof}

\section{Permuting fibers}\label{Fabian}

Direct analysis of the fibers of the discrete derivative of a power function with fractional exponent can be difficult.
This section shows that we can compose the discrete derivative with permutations to obtain more tractable functions.
\cref{Genevieve} has the general theory of how composition with permutations transforms fibers.
\cref{Eustace} then considers the particular transformations that simply our task for the discrete derivative of a power function with a fractional exponent.

\subsection{General principles}\label{Genevieve}
The following two lemmas show how composition with a permutation transforms the fibers of a function, first for composition on the left, and then for composition on the right.
\begin{lemma}\label{Fedor}
Let $A$ and $B$ be sets, let $\pi$ be a permutation of $B$, let $g\colon A \to B$ and $f=\pi \circ g$.  Then for each $b \in B$ we have $f^{-1}(\{b\})=g^{-1}(\{\pi^{-1}(b)\})$, so that $|f^{-1}(\{b\})|=|g^{-1}(\{\pi^{-1}(b)\})|$.  Thus, the multiset of cardinalities of fibers of $f$ is the same as the multiset of cardinalities of fibers of $g$.
\end{lemma}
\begin{proof}
We know that $a \in f^{-1}(\{b\})$ if and only if $b=\pi(g(a))$, which is true if and only if $\pi^{-1}(b)=g(a)$, which is true if and only if $a \in g^{-1}(\{\pi^{-1}(b)\})$.  Since $\pi$ is a permutation of $B$, we know that $\pi^{-1}(b)$ runs through $B$ as $b$ runs through $B$, so the multiset of cardinalities of fibers of $f$ is the same as that of $g$.
\end{proof}
\begin{lemma}\label{Leonid}
Let $A$ and $B$ be sets, let $\pi$ be a permutation of $A$, let $g\colon A \to B$, and let $f=g \circ \pi$.  Then for each $b \in B$ we have $f^{-1}(\{b\})=\pi^{-1}(g^{-1}(\{b\}))$, so that $|f^{-1}(\{b\})|=|g^{-1}(\{b\})|$.  Thus, the multiset of cardinalities of fibers of $f$ is the same as the multiset of cardinalities of fibers of $g$.
\end{lemma}
\begin{proof}
We know that $a \in f^{-1}(\{b\})$ if and only if $b=g(\pi(a))$, which is true if and only if $\pi(a) \in g^{-1}(\{b\})$, which is true if and only if $a \in \pi^{-1}(g^{-1}(\{b\}))$.  Since $\pi$ is a permutation, the cardinalities of $\pi^{-1}(g^{-1}(\{b\}))$ and $g^{-1}(\{b\})$ are the same for each $b \in B$.
\end{proof}

\subsection{Fibers of the discrete derivative}\label{Eustace}
In this section, we transform the discrete derivative of a power function with selected permutations that lead to functions whose fibers are easier to analyze.  There are four transformations that one can apply in succession, although the second and third are only helpful when our power function has a fractional exponent, while the fourth only applies in odd characteristic.
The first two transformations (in Lemmas \ref{Aaron} and \ref{Beatrice}) are inspired by the work of Hertel and Pott (see the proof of their Theorem 10 in \cite{Hertel-Pott}).
First we describe a permutation that is used in some of the transformations.
\begin{lemma}\label{Xavier}
Let $F$ be a finite field and let $\pi\colon F \to F$ with $\pi(x)=x^{\cardF-2}+1$.
Then $\pi$ is a permutation of $F$ with $\pi(0)=1$ and $\pi(x)=x^{-1}+1$ for $x\not=0$.
\end{lemma}
\begin{proof}
The values of $\pi(0)$ and $\pi(x)$ for nonzero $x$ are clear, and since $x^{-1}+1$ can never equal $1$ and $x^{-1}+1=y^{-1}+1$ if and only if $x=y$, we see that $\pi$ is an injective function from a finite set into itself, hence a permutation.
\end{proof}
The following four lemmas describe four transformations involving composition with permutations that successively take the discrete derivative $\Delta f$ of our power function $f$ to more tractable forms, which we call $f_1$, $f_2$, $f_3$, and $f_4$.
The first transformation applies to all power functions.
\begin{lemma}\label{Aaron}
Let $F$ be a finite field, $d$ a positive integer, and $f\colon F \to F$ the power map $f(x)=x^d$.
Let $\pi\colon F \to F$ be defined by $\pi(x)=x^{\cardF-2}+1$.
Let $f_1 \colon F \to F$ be defined by $f_1(1)=1$ and $f_1(x)=(x^d-1)/(x-1)^d$ for $x\not=1$.
Then $\Delta f=f_1\circ \pi$.
\end{lemma}
\begin{proof}
We have $(f_1\circ\pi)(0)=f_1(1)=1=f(1)-f(0)=(\Delta f)(0)$.
For $x \in \Fu$, by \cref{Xavier} we have
\begin{align*}
(f_1\circ \pi)(x)
& = \frac{(x^{-1}+1)^d-1}{((x^{-1}+1)-1)^d} \\
& = (x+1)^d-x^d. \qedhere
\end{align*}
\end{proof}
The next two transformations are only helpful when the power function's exponent is in fractional form.
\begin{lemma}\label{Beatrice}
Let $F$ be a finite field and let $d_1$ and $d_2$ be positive integers with $\gcd(d_2,\cardFu)=1$.
Let $f_1$ be the function over $F$ defined in \cref{Aaron} with exponent $d$ set to be the fractional exponent $d_1/d_2$.
Let $\sigma\colon F \to F$ with $\sigma(x)=x^{d_2}$.
Then $\sigma$ is a permutation of $F$.
Let $f_2\colon F \to F$ with $f_2(1)=1$ and $f_2(x)=(x^{d_1}-1)^{d_2}/(x^{d_2}-1)^{d_1}$ for $x\not=1$.
Then $f_2$ is a defined function from $F$ to $F$ with $f_2=\sigma \circ f_1 \circ \sigma$.
\end{lemma}
\begin{proof}
Notice that the denominator in the definition of $f_2$ is $(\sigma(x)-1)^{d_1}$, which is zero if and only if $\sigma(x)=1$, which is true if and only if $x=1$ because $\sigma$ is a permutation by \cref{Herbert} since we assumed that $\gcd(d_2,\cardFu)=1$.  This makes $f_2$ defined.
Notice that $(\sigma\circ f_1\circ\sigma)(1)=(\sigma\circ f_1)(1)=\sigma(1)=1=f_2(1)$.
On the other hand if $x \in F$ with $x\not=1$, then
\begin{align*}
(\sigma \circ f_1 \circ \sigma)(x)
& = \left(\frac{x^{d_2 d}-1}{(x^{d_2}-1)^d}\right)^{d_2} \\
& = \frac{(x^{d_1}-1)^{d_2}}{(x^{d_2}-1)^{d_1}} \\
& = f_2(x).\qedhere
\end{align*}
\end{proof}
\begin{lemma}\label{Celia}
Let $F$ be a finite field and let $d_1$ and $d_2$ be positive integers with $\gcd(d_2,\cardFu)=1$.
Let $f_2\colon F \to F$ be the function from \cref{Beatrice}.
Let $\pi\colon F \to F$ be defined by $\pi(x)=x^{\cardF-2}+1$.
Let $f_3\colon F \to F$ be given by 
\[
f_3(x)=\frac{((x+1)^{d_1}-x^{d_1})^{d_2}}{((x+1)^{d_2}-x^{d_2})^{d_1}}.
\]
Then $f_3$ is a defined function from $F$ to $F$ with $f_3 =f_2 \circ \pi$.
\end{lemma}
\begin{proof}
Note that $x\mapsto x^{d_2}$ is a permutation of $F$ by \cref{Herbert} since $\gcd(d_2,\cardFu)=1$.  Thus, the denominator in the definition of $f_3$ is nonzero for every $x \in F$, making $f_3$ a defined function.
We have $(f_2 \circ \pi)(0)=f_2(1)=1=f_3(0)$.
If $x \in \Fu$, then by \cref{Xavier} we have
\begin{align*}
(f_2 \circ \pi)(x)
& = \frac{((x^{-1}+1)^{d_1}-1)^{d_2}}{((x^{-1}+1)^{d_2}-1)^{d_1}} \\
& = \frac{((x+1)^{d_1}-x^{d_1})^{d_2}}{((x+1)^{d_2}-x^{d_2})^{d_1}}. \qedhere
\end{align*}
\end{proof}
The fourth and final transformation, which we describe next, is only possible in odd characteristic.
The terms $x+2$ and $x-2$ that appear in the formula for $f_4(x)$ make it easy to complete the square if one substitutes $x+x^{-1}$ for $x$ in powers of $x+2$ and $x-2$ (as is done in \cref{Rita}).
\begin{lemma}\label{Dorothy}
Let $F$ be a finite field of odd characteristic and let $d_1$ and $d_2$ be positive integers with $\gcd(d_2,\cardFu)=1$.
Let $f_3\colon F \to F$ be the function from \cref{Celia}.
Let $\tau \colon F \to F$ with $\tau(x)=(x-2)/4$.
Then $\tau$ is a permutation.
Let $f_4 \colon F \to F$ with
\[
f_4(x)=\frac{((x+2)^{d_1}-(x-2)^{d_1})^{d_2}}{((x+2)^{d_2}-(x-2)^{d_2})^{d_1}}.
\]
Then $f_4$ is a defined function from $F$ to $F$ with $f_4=f_3 \circ \tau$ and with $f_4(-x)=(-1)^{d_2-d_1} f_4(x)$ for every $x \in F$.
\end{lemma}
\begin{proof}
Note that $x\mapsto x^{d_2}$ is a permutation of $F$ by \cref{Herbert} since $\gcd(d_2,\cardFu)=1$.
Thus, the denominator in the definition of $f_4$ is nonzero for every $x \in F$ because $F$ is of odd characteristic.
This makes $f_4$ a defined function.
Note also that $\tau$ is a permutation because it has an inverse function $x\mapsto 4 x+2$.
For any $x \in F$, we have
\begin{align*}
(f_3 \circ \tau)(x)
& = \frac{\left(\left(\frac{x+2}{4}\right)^{d_1}-\left(\frac{x-2}{4}\right)^{d_1}\right)^{d_2}}{\left(\left(\frac{x+2}{4}\right)^{d_2}-\left(\frac{x-2}{4}\right)^{d_2}\right)^{d_1}} \\
& = f_4(x).
\end{align*}
If $x \in F$, then
\begin{align*}
f_4(-x)
& = \frac{\left((-x+2)^{d_1}-(-x-2)^{d_1}\right)^{d_2}}{\left((-x+2)^{d_2}-(-x-2)^{d_2}\right)^{d_1}} \\
& = \frac{\left((-1)^{d_1+1} \left((x+2)^{d_1}-(x-2)^{d_1}\right)\right)^{d_2}}{\left((-1)^{d_2+1}\left((x+2)^{d_2}-(x-2)^{d_2}\right)\right)^{d_1}} \\
& = (-1)^{d_2-d_1} f_4(x). \qedhere
\end{align*}
\end{proof}
\begin{remark}
The maps $\pi$, $\sigma$, and $\tau$ that appear in Lemmas \ref{Aaron}--\ref{Dorothy} are all permutations, so that by Lemmas \ref{Fedor} and \ref{Leonid}, the functions $\Delta f$, $f_1$, $f_2$, $f_3$, and $f_4$ from Lemmas \ref{Aaron}--\ref{Dorothy} all have the same multiset of cardinalities of fiber sizes.  Thus, the differential spectrum of $\Delta f$ (which is the discrete derivative of a power function) can be determined by examining $f_1$, $f_2$, $f_3$, or $f_4$.
\end{remark}
\begin{example}\label{Sven}
Let $F$ be the finite field of order $3^3$, and let $f\colon F \to F$ be the power function $f(x)=x^d$, where $d$ is the exponent $(3^3+1)/(3^2+1)=28/10=14/5$, where we note that $\gcd(5,3^3-1)=1$.  In an example in the Introduction, we saw that $f(x)=x^8$.  This section has shown that all the following functions have the same multiset of fiber sizes:
\begin{align*}
(\Delta f)(x) & = (x+1)^8-x^8 \\
f_1(x)        & = \begin{cases} 1 & \text{if $x=1$} \\ \frac{x^8-1}{(x-1)^8} & \text{otherwise} \end{cases} \\
f_2(x)        & = \begin{cases} 1 & \text{if $x=1$} \\ \frac{(x^{14}-1)^5}{(x^5-1)^{14}} & \text{otherwise} \end{cases} \\
f_3(x)        & = \frac{\left((x+1)^{14}-x^{14}\right)^5}{\left((x+1)^5-x^5\right)^{14}} \\
f_4(x)        & = \frac{\left((x+2)^{14}-(x-2)^{14}\right)^5}{\left((x+2)^5-(x-2)^5\right)^{14}}.
\end{align*}
\end{example}

\section{A double covering of a finite field}\label{David}

Permuting fibers of our discrete derivative using the ideas of \cref{Fabian} is not of itself sufficient for the analysis in \cref{Victoria}; we must go beyond permutations.
When one wants to track the sizes of transformed fibers, the best-behaved transformations are maps whose fibers are all the same size.
The simplest of these are bijections (all fibers of size $1$) such as permutations, but in this work we shall need to venture further and use a double covering (all fibers of size $2$).
The double covering we use is the finite field analogue of the double covering of the real line $\R$ by the map $x \mapsto x+x^{-1}$ whose domain is the disjoint union of $\R^*$ and the complex unit circle.
In \cref{Grace}, we define the analogous sets for the finite field case, and in \cref{Donald}, we define the analogous map, prove that it is a double covering, and show how it affects fibers when we compose with it.

\subsection{Half field, conjugation, and the unit circle}\label{Grace}

Here we develop the analogy between finite fields and the real and complex fields.
A finite field of even degree over its prime subfield can be thought of as an analogue of $\C$, and then its subfield of index $2$ is analogous to $\R$.
We develop a terminology and notation for this.
\begin{definition}[half field]
If $F$ is a finite field of characteristic $p$ and order $q=p^n$ with $n$ even, then the {\it half field of $F$}, denoted $H_F$, is the unique subfield of $F$ with $[F:H_F]=2$, i.e., the subfield of order $p^{n/2}$.
\end{definition}
Then the nontrivial automorphism of a field over its half field is analogous to complex conjugation.
\begin{definition}[conjugation]
If $F$ is a finite field of characteristic $p$ and order $q=p^n$ with $n$ even, then the {\it conjugation map on $F$} is the power map $x\mapsto x^{p^{n/2}}$ of $F$.  This is an automorphism of $F$ of order $2$ whose fixed field is the half field $H_F$.  The {\it conjugate of $x$} is $x^{p^{n/2}}$, and is often denoted $\bar{x}$.  Thus, if we extend $\bar{x}$ to mean $x^{p^{n/2}}$ for all $x \in \acF$, then $x \in H_F$ if and only if $\bar{x}=x$ and $x \in F$ if and only if $\bar{\bar{x}}=x$.
\end{definition}
Then elements whose conjugates equal their inverses can be thought of as forming a unit circle.
\begin{definition}[unit circle]
If $F$ is a finite field of characteristic $p$ and order $q=p^n$ with $n$ even, then the {\it unit circle of $F$}, denoted $U_F$, is the set $\{x \in F: x^{p^{n/2}}=x^{-1}\}$.  Equivalently, it is the set of elements in $F$ whose conjugates (using the conjugation map on $F$) equal their inverses.  It is also the unique subgroup of $\Fu$ (or $\acFu$) of order $p^{n/2}+1$, and it is a cyclic subgroup generated by $\alpha^{p^{n/2}-1}$ if $\alpha$ is a primitive element of $F$.
\end{definition}
Just as in the characteristic $0$ case, the intersection of the unit circle and the half field consists of the two points $1$ and $-1$.
\begin{lemma}\label{Lenny}
Let $F$ be a finite field of characteristic $p$ and order $p^n$ with $n$ even.  Then $H_F\cap U_F=\{1,-1\}$.
\end{lemma}
\begin{proof}
Let $x\mapsto \bar{x}$ denote the conjugation map on $F$.
Then $H_F\cap U_F$ is the set of all $x \in F$ with both $\bar{x}=x$ and $\bar{x}=x^{-1}$, so that a necessary condition for being in the intersection is that $x=x^{-1}$, i.e., $x\in\{-1,1\}$.
This condition is also clearly sufficient since $1$ and $-1$ are self-conjugate and self-inverse.
\end{proof}

\subsection{The fiber doubling map}\label{Donald}

In this section, we introduce our double covering.
First we define the map.
\begin{definition}[Fiber doubling map]\label{Freddie}
Let $F$ be a finite field, $E$ the quadratic extension of $F$, and $U_E$ the unit circle of $E$.
We let $\du=\left\{(x,S): S \in \{\Fu,U_E\}, x \in S\right\}$ be the disjoint union of $\Fu$ and $U_E$.
Then the {\it fiber doubling map for $F$ is} $\dc \colon \du \to F$ with $\dc(x,S)=x+x^{-1}$ for every $(x,S) \in \du$.
\end{definition}
This validity of this definition actually requires us to know that $x+x^{-1} \in F$ when $x \in U_E$, which is clear since $x^{\cardF}=x^{-1}$ for every $x \in U_E$, so that $(x+x^{-1})^{\cardF}=x^{\cardF}+x^{-\cardF}=x^{-1}+x=x+x^{-1}$.
In fact, more will be said about this map in \cref{Gerald}, after we prove a preliminary lemma.
\begin{lemma}\label{Simon}
Let $K$ be a field and let $x,y \in K^*$.
Then $x+x^{-1}=y+y^{-1}$ if and only if $y \in \{x,x^{-1}\}$.
\end{lemma}
\begin{proof}
Notice that $x+x^{-1}=y+y^{-1}$ if and only if
\[
(x-y)\left(1-\frac{1}{x y}\right)=0.
\]
Since $K$ is a field, this is true if and only if $x-y=0$ or $x y=1$.
\end{proof}
Now we show that our fiber doubling map is indeed a double covering.
\begin{lemma}\label{Gerald}
Let $F$, $E$, $U_E$, $\du$, and $\dc$ be as in \cref{Freddie}.  Then $\dc$ is a $2$-to-$1$ map from $\du$ onto $F$.
\end{lemma}
\begin{proof}
We should keep in mind that since $F=H_E$, \cref{Lenny} shows that $\Fu\cap U_E=\{1,-1\}$.
Suppose that $\dc(x,S)=\dc(y,T)$ for two distinct elements $(x,S)$ and $(y,T)$ of $\du$.
Then $x+x^{-1}=y+y^{-1}$, so \cref{Simon} tells us that $y \in \{x,x^{-1}\}$.
If $y=x$, then we must have $S\not=T$ to maintain the distinctness of $(x,S)$ and $(y,T$), and then $y=x$ lies in $\Fu\cap U_E=\{1,-1\}$, and thus $y=x^{-1}$.
So in fact we must have $y=x^{-1}$ in all cases.
When $x \in \{1,-1\}$, then $y=x^{-1}=x$, so then $T$ must be whichever element of $\{\Fu,U_E\}$ is not $S$, but if $x \not\in\{1,-1\}$ then each of $x$ and $y=x^{-1}$ lies in one and only one of $\Fu$ or $U_E$ (and they lie in the same one, since $\Fu$ and $U_E$ are groups under multiplication), so that $S=T$.
Thus, we have shown that if $(x,S)$ lies in some fiber of $\dc$, then the only other element that could lie in the same fiber is $(x^{-1},T)$ where we must have $T\not=S$ when $x \in\{1,-1\}$ and we must have $T=S$ otherwise.
Thus, no fiber of $\dc$ has more than two elements, but since the domain of $\dc$ has twice as many elements as the codomain, this forces every fiber to have exactly two elements, making $\dc$ surjective.
\end{proof}
The following corollary shows how composition with the fiber doubling map transforms fibers.
\begin{corollary}\label{Edgar}
Let $F$, $E$, $U_E$, $\du$, and $\dc$ be as in \cref{Freddie}, and let $f\colon F \to F$ be a function.  Then for every $c \in F$, we have $|f^{-1}(\{c\})|=|(f\circ\dc)^{-1}(\{c\})|/2$, and $\dc\left((f\circ\dc)^{-1}(\{c\})\right)=f^{-1}(\{c\})$.
\end{corollary}
\begin{proof}
Since \cref{Gerald} tells us that $\dc$ is a $2$-to-$1$ map from $\du$ onto $F$, we note that $(f\circ\dc)^{-1}(\{c\})=\dc^{-1}\left(f^{-1}(\{c\})\right)$ will be twice as large as $f^{-1}(\{c\})$.
Furthermore, $\dc\left((f\circ\dc)^{-1}(\{c\})\right)=\dc\left(\dc^{-1}\left(f^{-1}(\{c\})\right)\right)=f^{-1}(\{c\})$ because $\dc$ is surjective.
\end{proof}

\section{Analysis of the fibers of the function of Lemma \ref{Dorothy} via the fiber doubling map}\label{Helen}

In \cref{Fabian}, we showed that when the discrete derivative $\Delta f$ of a power function $f$ with a fractional exponent is composed with selected permutations, one obtains the more tractable function $f_4$ described in \cref{Dorothy}.
This is not enough on its own to obtain the results of \cref{Victoria}, so in this section, we compose $f_4$ with the double covering $\dc$ of \cref{Donald} to further simply the functional form.
The three smaller sections within this section analyze $f_4\circ\dc$ in more and more specific cases.
First, \cref{Beverly} considers $f_4\circ\dc$ arising from $\Delta f$ where $f$ is a power function over a field of odd characteristic $p$ and order $q=p^n$ that has an exponent of the form $(p^j+1)/(p^k+1)$.
Then, \cref{Evelyn} specializes to the case where $j=n$, so that the exponent is of the form $(q+1)/(p^k+1)$.
Finally, \cref{Juliette} further restricts to the precise conditions of \cref{Victoria} by setting the characteristic $p$ to $3$ and insisting that $n$ be odd, $k$ be even, and $\gcd(k,n)=1$.

\subsection{When $d_1=(p^j+1)/2$ and $d_2=(p^k+1)/2$ in $x\mapsto x^{d_1/d_2}$ over a field of odd characteristic $p$}\label{Beverly}

In this section, we first algebraically simplify $f_4\circ\dc$ in \cref{Katherine}, and then demonstrate its symmetries in \cref{Thomas}.
\begin{proposition}\label{Katherine}
Let $F$ be a finite field of odd characteristic $p$, let $j$ and $k$ be nonnegative integers, let $d_1=(p^j+1)/2$ and $d_2=(p^k+1)/2$, and suppose that $\gcd(d_2,\cardFu)=1$.
Let $\theta=1$ if $p\equiv\pm 1\pmod{8}$ or $j \equiv k \pmod{2}$; otherwise let $\theta=-1$.
Let $f_4$ be as defined in \cref{Dorothy}.
Let $E$ be the quadratic extension of $F$, let $U_E$ be the unit circle of $E$, and let $\dc\colon\du \to F$ be the fiber doubling map.
If $(x,S) \in \du$, then $x^{p^k}+x\not=0$ and 
\[
(f_4\circ\dc)(x,S) = \theta \frac{(x^{p^j}+x)^{d_2}}{(x^{p^k}+x)^{d_1}}.
\]
\end{proposition}
\begin{proof}
The condition $\gcd(d_2,\cardFu)=1$ makes $f_4$ defined on all of $F$ by \cref{Dorothy}.
Let $(x,S) \in \du$.
We have
\begin{align*}
(f_4\circ\dc)(x,S)
& = f_4(x+x^{-1}) \\
& = \frac{((x+x^{-1}+2)^{d_1}-(x+x^{-1}-2)^{d_1})^{d_2}}{((x+x^{-1}+2)^{d_2}-(x+x^{-1}-2)^{d_2})^{d_1}},
\end{align*}
and we know the denominator does not vanish since $x+x^{-1} \in F$ and $x\mapsto x^{d_2}$ is a permutation of $F$ by \cref{Herbert} because $\gcd(d_2,\cardFu)=1$.
\cref{Rita} shows that the denominator is equal to $2^{d_1} x^{-d_1 d_2} (x^{p^k}+x)^{d_1}$, so we know that $x^{p^k}+x\not=0$.
\cref{Rita} also shows that the numerator is equal to $2^{d_2} x^{-d_1 d_2} (x^{p^j}+x)^{d_2}$, and so
\[
(f_4\circ\dc)(x,S) = 2^{d_2-d_1} \frac{(x^{p^j}+x)^{d_2}}{(x^{p^k}+x)^{d_1}}.
\]
Notice that $d_2-d_1 = (p^k-p^j)/2=(p^k-1)/2-(p^j-1)/2$, which by \cref{Molly} is congruent to $(k-j)(p-1)/2$ modulo $p-1$.
Thus, if $j$ and $k$ have the same parity, then $2^{d_2-d_1}=1$ by Fermat's little theorem.
So assume $j$ and $k$ are of opposite parity henceforth, and then $2^{d_2-d_1} = 2^{(p-1)/2}$, which is $1$ if $2$ is a quadratic residue in $\Fpu$ or $-1$ if $2$ is a quadratic nonresidue in $\Fpu$.
By the second supplement to the law of quadratic reciprocity, $2$ is a quadratic residue in $\Fpu$ if $p\equiv\pm 1\pmod{8}$, but is a quadratic nonresidue if $p\equiv\pm 3\pmod{8}$.
\end{proof}
The following result shows the symmetries of $f_4\circ\dc$ under additive and multiplicative inversion.
\begin{lemma}\label{Thomas}
Let $F$ be a finite field of odd characteristic $p$, let $j$ and $k$ be nonnegative integers, let $d_1=(p^j+1)/2$ and $d_2=(p^k+1)/2$, and suppose that $\gcd(d_2,\cardFu)=1$.
Let $\kappa=-1$ if $p\equiv 3\pmod{4}$ and $j$ is odd; otherwise let $\kappa=1$.
Let $f_4$ be as defined in \cref{Dorothy}.
If $x \in F$, then
\[
f_4(-x) = \kappa f_4(x).
\]
Let $E$ be the quadratic extension of $F$, let $U_E$ be the unit circle of $E$, and let $\dc\colon\du \to F$ be the fiber doubling map.
If $(x,S) \in \du$, then
\begin{align*}
(f_4\circ\dc)(1/x,S) & = (f_4\circ\dc)(x,S) \\
(f_4\circ\dc)(-x,S) & = \kappa(f_4\circ\dc)(x,S) \\
(f_4\circ\dc)(-1/x,S) & = \kappa(f_4\circ\dc)(x,S).
\end{align*}
\end{lemma}
\begin{proof}
Let $x \in F$.  By \cref{Dorothy}, we have $f_4(-x)=(-1)^{d_1-d_2}f_4(x)$.
Notice that $d_2$ is always odd because $\gcd(d_2,\cardFu)=1$.
\cref{Francis} tells us that $d_1=(p^j+1)/2$ is even if and only both $p\equiv 3 \pmod{4}$ and $j$ is odd.
Thus, $(-1)^{d_1-d_2}=\kappa$, so that $f_4(-x)=\kappa f_4(x)$.

Let $(x,S) \in \du$.
Since $\dc(1/x,S)=\dc(x,S)$, it is clear that $(f_4\circ\dc)(1/x,S) = (f_4\circ\dc)(x,S)$.
We have $(-x,S) \in\du$ since $\Fu$ and $U_E$ are groups under multiplication and contain $-1$.
Then $\dc(-x,S)=-\dc(x,S)$ and so we have
\begin{align*}
(f_4\circ\dc)(-x,S)
& = f_4(-\dc(x,S)) \\
& = \kappa f_4(\dc(x,S)),
\end{align*}
where the last equality has already been established earlier in this proof.
The final relation that we are to prove follows from the two that immediately precede it.
\end{proof}
\begin{example}
Let $F$ be the finite field of order $3^3$, and let $f\colon F \to F$ be the power function $f(x)=x^d$, where $d$ is the exponent $(3^3+1)/(3^2+1)=28/10=14/5$, where we note that $\gcd(5,3^3-1)=1$.
In an example in the Introduction, we saw that $f(x)=x^8$.
Let $f_4$ be as defined in \cref{Dorothy}.
Then \cref{Sven} shows that
\[
f_4(x)          = \frac{\left((x+2)^{14}-(x-2)^{14}\right)^5}{\left((x+2)^5-(x-2)^5\right)^{14}}.
\]
Let $E$ be the quadratic extension of $F$, that is, the finite field of order $3^6$, let $U_E$ be the unit circle of $E$, and let $\dc \colon \Fu \sqcup U_E \to F$ be the fiber doubling map from \cref{Freddie}.
Then \cref{Katherine} shows that for $(x,S) \in \Fu\sqcup U_E$, we have
\[
(f_4\circ\dc)(x,S) = -\frac{(x^{27}+x)^{5}}{(x^9+x)^{14}}.
\]
\cref{Thomas} shows that $f_4(-x)=-f_4(x)$ for $x \in F$ and that $(f_4\circ\dc)(1/x,S)=(f_4\circ\dc)(x,S)$ and $(f_4\circ\dc)(-x,S)=(f_4\circ\dc)(-1/x,S)=-(f_4\circ\lambda)(x,S)$ for $(x,S) \in \Fu\sqcup U_E$, which are facts that can also be checked directly using the explicit forms above.
\end{example}

\subsection{When $d_1=(q+1)/2$ and $d_2=(p^k+1)/2$ in $x\mapsto x^{d_1/d_2}$ over a field of odd characteristic $p$ and order $q$}\label{Evelyn}
In this section, we restrict the numerator of our power function's fractional exponent $(p^j+1)/(p^k+1)$ to have $p^j=q$, where $q$ is the order of the field on which the power function is defined.
This specialization allows for further algebraic simplification of $f_4\circ\dc$ to be seen in \cref{Irene}, which also looks at the square of this composed function, which plays a role in later calculations.
Then \cref{Uriah} analyzes the fibers of a function that appears as a term in part \ref{Julia} of \cref{Irene}, and \cref{Madeleine} finishes with the determination of a special fiber of $f_4\circ\dc$.
\begin{proposition}\label{Irene}
Let $F$ be a finite field of odd characteristic $p$ and order $q=p^n$, let $k$ be a nonnegative integer, let $d_1=(p^n+1)/2$, $d_2=(p^k+1)/2$, $e_2=(p^k-1)/2$, and suppose that $\gcd(d_2,\cardFu)=1$.
Let $f_4$ be as defined in \cref{Dorothy}.
Let $E$ be the quadratic extension of $F$, let $U_E$ be the unit circle of $E$, and let $\dc\colon\du \to F$ be the fiber doubling map.
\begin{enumerate}[label=(\roman*)]
\item\label{Julius}
Let $x \in \Fu$ and let $\mu=1$ if $q\equiv\pm 1\pmod{8}$; otherwise let $\mu=-1$.
Then $x^{e_2}+x^{-e_2}\not=0$ and 
\begin{align*}
(f_4\circ\dc)(x,\Fu) & = 2 \mu x^{(q-1)/2} \left(\frac{1}{x^{e_2}+x^{-e_2}}\right)^{(q+1)/2} \not=0 \text{ and}\\
\left((f_4\circ\dc)(x,\Fu)\right)^2 & = \left(\frac{2}{x^{e_2}+x^{-e_2}}\right)^2 = 1-\left(\frac{x^{e_2}-x^{-e_2}}{x^{e_2}+x^{-e_2}}\right)^2\not=0.
\end{align*}
\item\label{Julia}
Let $x \in U_E$ and let $\theta=1$ if $p\equiv\pm 1\pmod{8}$ or $n \equiv k \pmod{2}$; otherwise let $\theta=-1$.
Then $x^{p^k}+x\not=0$ and 
\begin{align*}
(f_4\circ\dc)(x,U_E) & = \theta\frac{(x+x^{-1})^{(p^k+1)/2}}{(x^{p^k}+x)^{(q+1)/2}} \text{ and}\\
\left((f_4\circ\dc)(x,U_E)\right)^2 & = 1+ \left(\frac{x^{p^k+1}-1}{x^{p^k}+x}\right)^2.
\end{align*}
\end{enumerate}
\end{proposition}
\begin{proof}
Since $\gcd(d_2,\cardFu)=1$ we know that $d_2$ is odd and \cref{Dorothy} tells us that $f_4$ is defined.
We adopt the definition of $\theta$ given in part \ref{Julia} of this proposition throughout this proof.
Then \cref{Katherine} with $j=n$ tells us that if $(x,S) \in \du$, then $x^{p^k}+x\not=0$ and 
\begin{equation}\label{Manfred}
(f_4\circ\dc)(x,S) = \theta\frac{(x^q+x)^{d_2}}{(x^{p^k}+x)^{(q+1)/2}}.
\end{equation}

When $S=\Fu$ in \eqref{Manfred}, we have $x \in \Fu$, so $x^q=x$, and we obtain
\[
(f_4\circ\dc)(x,\Fu) = \theta\frac{(x+x)^{d_2}}{(x^{p^k}+x)^{(q+1)/2}}.
\]
Since we know that $x^{p^k}+x=x^{d_2} (x^{e_2}+x^{-e_2})$ is nonzero, this shows that $x^{e_2}+x^{-e_2}\not=0$ and so
\begin{align*}
(f_4\circ\dc)(x,\Fu) 
& = \theta  \frac{(2 x)^{d_2}}{\left(x^{d_2} (x^{e_2}+x^{-e_2})\right)^{(q+1)/2}} \\
& = 2^{d_2} \theta  x^{-d_2(q-1)/2} \left(\frac{1}{x^{e_2}+x^{-e_2}}\right)^{(q+1)/2} \\
& = 2^{d_2} \theta  x^{(q-1)/2} \left(\frac{1}{x^{e_2}+x^{-e_2}}\right)^{(q+1)/2},
\end{align*}
where the last equality uses the fact that $x \in \Fu$ and $d_2$ is odd.
Let $r$ be a square root of $2$ in the algebraic closure of $F$.
If $p\equiv\pm 1\pmod{8}$, then $r\in\Fp$ by the second supplement to the law of quadratic reciprocity, and then $r^p=r$; otherwise $r\in\F_{p^2}\smallsetminus\Fp$, and then $r^p=-r$.
Thus, $2^{d_2} = r^{p^k+1} = r^2=2$ if $p\equiv\pm 1\pmod{8}$ or if $k$ is even; otherwise $2^{d_2} = r^{p^k}r=-r^2=-2$.
From this and the definition of $\theta$, we deduce that $2^{d_2} \theta =2$ if $p\equiv\pm 1\pmod{8}$ or $n$ is even; otherwise $2^{d_2} \theta =-2$.
That is, $2^{d_2} \theta = 2 \mu$, and we obtain the desired formula for $(f_4\circ\dc)(x,\Fu)$.
The term $1/(x^{e_2}+x^{-e_2})$ in that formula cannot be zero, so that both it and $x$ lie in $\Fu$, and so $\left((f_4\circ\dc)(x,\Fu)\right)^2 =4 x^{q-1} (x^{e_2}+x^{-e_2})^{-(q+1)} = \left(2/(x^{e_2}+x^{-e_2})\right)^2$.
We finish the proof of part \ref{Julius} by observing that
\begin{align*}
\left(\frac{2}{x^{e_2}+x^{-e_2}}\right)^2
& = \frac{(x^{e_2}+x^{-e_2})^2-(x^{e_2}-x^{-e_2})^2}{(x^{e_2}+x^{-e_2})^2} \\
&=1-\left(\frac{x^{e_2}-x^{-e_2}}{x^{e_2}+x^{-e_2}}\right)^2.
\end{align*}

When $S=U_E$ in \eqref{Manfred}, we have $x \in U_E$, so $x^q=x^{-1}$, and we obtain
\[
(f_4\circ\dc)(x,U_E) = \theta\frac{(x+x^{-1})^{(p^k+1)/2}}{(x^{p^k}+x)^{(q+1)/2}}.
\]
For the rest of this proof, we use the notation $\tilde{z}$ as a shorthand for $z^{p^k}$.
Then 
\begin{align*}
((f_4\circ\dc)(x,U_E))^2
& = \frac{(x+x^{-1})^{p^k+1}}{(\tilde{x}+x)^{q+1}} \\
& = \frac{(x+x^{-1}) (\tilde{x}+\tilde{x}^{-1})}{(\tilde{x}+x)(\tilde{x}^q+x^q)} \\
& = \frac{(x+x^{-1}) (\tilde{x}+\tilde{x}^{-1})}{(\tilde{x}+x)(\tilde{x}^{-1}+x^{-1})} && \text{since $x \in U_E$} \\
& = \frac{(x^2+1) (\tilde{x}^2+1)}{(\tilde{x}+x)^2} \\
& = \frac{(x+\tilde{x})^2+(x\tilde{x}-1)^2}{(\tilde{x}+x)^2} \\
& = 1+ \left(\frac{x\tilde{x}-1}{\tilde{x}+x}\right)^2,
\end{align*}
which is the formula we were to prove.
\end{proof}
\begin{example}
Let $F$ be the finite field of order $3^3$, and let $f\colon F \to F$ be the power function $f(x)=x^d$, where $d$ is the exponent $(3^3+1)/(3^2+1)=28/10=14/5$, where we note that $\gcd(5,3^3-1)=1$.
In an example in the Introduction, we saw that $f(x)=x^8$.
Let $f_4$ be as defined in \cref{Dorothy}.
Then \cref{Irene} shows that for $x \in \Fu$, we have
\begin{align*}
(f_4\circ\dc)(x,\Fu) & = x^{13} \left(\frac{1}{x^{4}+x^{-4}}\right)^{14} \not=0 \text{ and}\\
\left((f_4\circ\dc)(x,\Fu)\right)^2 & = \left(\frac{2}{x^{4}+x^{-4}}\right)^2 = 1-\left(\frac{x^{4}-x^{-4}}{x^{4}+x^{-4}}\right)^2\not=0,
\end{align*}
but if $x \in U_E$, then 
\begin{align*}
(f_4\circ\dc)(x,U_E) & = -\frac{(x+x^{-1})^{5}}{(x^{9}+x)^{14}}\text{ and}\\
\left((f_4\circ\dc)(x,U_E)\right)^2 & = 1+ \left(\frac{x^{10}-1}{x^{9}+x}\right)^2.
\end{align*}
\end{example}
The following result does much of the heavy work in our later analysis.
Its function $h$ appears as the main term in the expression for $\left((f_4\circ\dc)(x,U_E)\right)^2$ in part \ref{Julia} of \cref{Irene}, and tracking the fibers of $h$ (and its square root $g$) will eventually enable a complete determination of the fibers of $f_4\circ\dc$ in \cref{Edward} in \cref{Juliette}.
\begin{proposition}\label{Uriah}
Let $F$ be a finite field of odd characteristic $p$ and order $q=p^n$ with $n$ odd, and let $\eta$ be the extended quadratic character of $F$.
Let $E$ be the quadratic extension of $F$, let $i$ and $-i$ be the primitive fourth roots of unity in $E$, let $U_E$ be unit circle of $E$, and let $x\mapsto \sqrt{x}$ be any function from $E$ to $E^{\alg}$ that maps each $x \in E$ to one of its square roots.
Let $k$ be a nonnegative even integer, let $B=\{x \in F: \eta(x)\not=1, \eta(x+1)\not=-1\}$, and let $\Gamma=\{x \in E: x^2 \in B\}$.
Then there are surjective maps $g\colon U_E \to \Gamma$ and $h\colon U_E \to B$ with $g(x)=(x^{p^k+1}-1)/(x^{p^k}+x)$ and $h(x)=g(x)^2$.
Let
\begin{align*}
r_1(x) & = \prod_{j=0}^{n-1} \left((-1)^j x^{3^{j k}}+1\right), \\
r_2(x) & = \prod_{j=0}^{n-1} \left((-1)^j x^{3^{j k}}-1\right)\text{, and} \\
r_3(x) & = \frac{r_1(x)+r_2(x)}{r_1(x)-r_2(x)}
\end{align*}
be in $\Fthree(x)$.
Then there is a surjective map $u\colon \Gamma \to U_E$ with
\[
u(\gamma)=\begin{cases}
\gamma & \text{if $\gamma \in \{i,-i\}$} \\
i r_3(-i\gamma)+\sqrt{1-r_3(-i\gamma)^2} & \text{otherwise,}
\end{cases}
\]
where $r_1(-i \gamma)\not=r_2(-i \gamma)$ when $\gamma\in\Gamma\smallsetminus\{i,-i\}$, so that $r_3(-i \gamma)$ is always defined in these cases.
We have
\[
-\frac{1}{u(\gamma)}=\begin{cases}
\gamma & \text{if $\gamma \in \{i,-i\}$} \\
i r_3(-i \gamma)-\sqrt{1-r_3(-i \gamma)^2} & \text{otherwise}
\end{cases}
\]
for every $\gamma\in\Gamma$.
Furthermore, for every $\gamma \in \Gamma$, we have $g^{-1}(\{\gamma\})=\{u(\gamma),-1/u(\gamma)\}$; and for every $\beta \in B$, we have
\[
h^{-1}(\{\beta\})=\{u(\sqrt{\beta}),-u(\sqrt{\beta}),1/u(\sqrt{\beta}),-1/u(\sqrt{\beta})\}.
\]
\end{proposition}
\begin{proof}
We use the notation $\tilde{z}$ as a shorthand for $z^{p^k}$ in this proof.
Note that $U_E$ is a group under multiplication and is of even order, so it contains $-1$.
Thus, whenever $x \in U_E$, we also have $-x$, $1/x$, and $-1/x$ in $U_E$.

We break the rest of the proof into steps to prove the various claims.
\begin{step}
{\bf The rules for $g$ and $h$ map $U_E$ into $E$:}
Since $n$ is odd and $k$ is even, \cref{Clementine}\ref{Arthur} tells us that $\gcd((p^k+1)/2,p^n-1)=1$, so we may apply \cref{Irene}\ref{Julia}, which tells us that $x^{p^k}+x\not=0$ for $x \in U_E$, and so the rules for $g$ and $h$ do at least give defined outputs in $E$.
\end{step}
\begin{step}
{\bf Symmetries of $g$:}
It will be important to note that for every $x \in U_E$ we have
\begin{align}\begin{split}\label{Priscilla}
g(-x)     & = -g(x) \\
g(x^{-1})  & = -g(x) \\ 
g(-x^{-1}) & =  g(x).
\end{split}\end{align}
The first identity is clear since our formula gives $g(x)$ as the quotient of polynomials where all terms in the numerator (resp., denominator) are of  even (resp., odd) degree.
The second identity of \eqref{Priscilla} is straightforward to check:
\[
g(x^{-1}) = \frac{x^{-1} \tilde{x}^{-1}-1}{\tilde{x}^{-1}+x^{-1}} = \frac{1-x \tilde{x}}{x+\tilde{x}} = -g(x),
\]
and then the third identity of \eqref{Priscilla} follows from the first and second.
\end{step}
\begin{step}
{\bf The rules for $h$ and $g$ map into $B$ and $\Gamma$, respectively:}
Let $x \in U_E$.
Since $0$ is clearly in $B$, we are finished if $h(x)=0$.
Thus, we assume that $h(x)\not=0$ (and so $g(x)\not=0$), and calculate
\begin{align*}
h(x)^{(q-1)/2} 
& = g(x)^{q-1} \\
& = g(x)^q g(x)^{-1} \\
& = g(x^q) g(x)^{-1} \\
& = g(x^{-1}) g(x)^{-1} && \text{since $x \in U_E$}\\
& = -1 && \text{by \eqref{Priscilla}}.
\end{align*}
Thus, $h(x)^{q-1}=1$, proving that $h(x) \in \Fu$, but $h(x)^{(q-1)/2}=-1$ shows that $\eta(h(x))=-1$.
Furthermore, \cref{Irene}\ref{Julia} shows that $h(x)+1$ is the square of some element of $F$, and so $\eta(h(x)+1)\not=-1$.  Thus, $h(x) \in B$.
So $h$ maps $U_E$ into $B$, and therefore $g$ maps $U_E$ into $\Gamma$.
\end{step}
\begin{step}\label{Giselle}
{\bf The fiber of $g$ containing $x$ is $\{x,-1/x\}$:}
If $x \in U_E$, then \eqref{Priscilla} tells us that the fiber of $g$ that contains $x$ must also contain $-1/x$.
Suppose that this fiber also contains some $y\not\in\{x,-1/x\}$ to show contradiction.
Then $y \in U_E$ since $U_E$ is the domain of $g$, and $g(x)=g(y)$, so that
\[
\frac{x\tilde{x}-1}{\tilde{x}+x} = \frac{y\tilde{y}-1}{\tilde{y}+y},
\]
so that
\[
(x\tilde{x}-1) (\tilde{y}+y) = (y\tilde{y}-1) (\tilde{x}+x),
\]
which rearranges to
\[
(\tilde{x} \tilde{y}+1)(x-y)=-(x y+1) (\tilde{x}-\tilde{y}),
\]
in other words,
\[
(x y+1)^{p^k} (x-y) = - (x y+1) (x-y)^{p^k},
\]
and since we are assuming $y \not\in \{x,-1/x\}$, we obtain
\begin{equation}\label{Wolfgang}
\left(\frac{x y+1}{x- y}\right)^{p^k-1} = -1.
\end{equation}
Since $x,y \in U_E \subseteq E$ and since $y\not \in \{x,-1/x\}$, we know that $(x y+1)/(x-y) \in \Eu$.
Since $\Eu$ is a cyclic group of order $p^{2 n}-1$, we see that $\left((x y+1)/(x- y)\right)^{p^k-1}$ lies in the unique subgroup $G$ of order $(p^{2 n}-1)/\gcd(p^{2 n}-1,p^k-1)=(p^{2 n}-1)/(p^{\gcd(2 n,k)}-1)$ (see \cref{Clementine}\ref{Hortense}) in $\Eu$.
On the other hand, $-1$ is an element of order $2$, and lies in $G$ if and only if $2 \mid (p^{2 n}-1)/(p^{\gcd(2 n,k)}-1)$, so \eqref{Wolfgang} implies $v_2(p^{2 n}-1) > v_2(p^{\gcd(2 n,k)}-1)$.
Since \cref{Francis} shows that $v_2(p^j-1)$ increases strictly as $v_2(j)$ increases, the last inequality implies that $v_2(2 n) > v_2(\gcd(2 n,k))$, which implies that $v_2(2 n) > v_2(k)$, which contradicts our assumption that $n$ is odd and $k$ is even.
Thus, the fiber of $g$ that contains $x$ is $\{x,-1/x\}$.
\end{step}
\begin{step}\label{Susan}
{\bf $g$ and $h$ are surjective:}
If $x \in E$, then (since we are working in odd characteristic) the set $\{x,-1/x\}$ has precisely two elements unless $x \in \{i,-i\}$.
Since $U_E$ is a multiplicative group of cardinality $q+1$, we see that $U_E$ either contains both $i$ and $-i$ when $q\equiv 3\pmod{4}$ or contains neither of them when $q\equiv 1 \pmod{4}$.
So, if $q \equiv 1 \pmod{4}$, \cref{Giselle} implies that $g$ has $(q+1)/2$ nonempty fibers, each of size $2$, while if $q \equiv 3 \pmod{4}$, then $g$ has $(q-1)/2$ fibers of size $2$, and $2$ fibers of size $1$, for a total of $(q+3)/2$ nonempty fibers.
The number of nonempty fibers of $g$ matches exactly the cardinality of its codomain $\Gamma$ (determined in \cref{James}), so $g$ is surjective.
Since every element of $B$ has a square root in $\Gamma$, this also shows that $h$ is surjective.
\end{step}
\begin{step}\label{Abigail}
{\bf The fiber of $g$ over $i$ or $-i$:}
\cref{James} tells us that $\Gamma$ contains $i$ and $-i$ if and only if $q\equiv 3 \pmod{4}$; this is also the necessary and sufficient condition for $U_E$ to contain $i$ and $-i$.
Note that $\tilde{i}=i$ since $x\mapsto \tilde{x}$ is an automorphism iterated an even number of times.
Thus, when $q\equiv 3 \pmod{4}$, we have $g(i)=(i\cdot i -1)/(i+i)=i$ and then $g(-i)=-i$ by \eqref{Priscilla}, so that \cref{Giselle} makes $g^{-1}(\{i\})=\{i,-1/i\}=\{i\}$ and $g^{-1}(\{-i\})=\{-i,1/i\}=\{-i\}$.
\end{step}
\begin{step}
{\bf An element of the fiber of $g$ over some $\gamma\in\Gamma\smallsetminus\{i,-i\}$ is fixed by a certain linear fractional transformation:}
Let $\gamma \in \Gamma\smallsetminus\{i,-i\}$.
Using the surjectivity of $g$ proved in \cref{Susan}, let $y \in U_E$ with $g(y)=\gamma$, that is, $(\tilde{y}y-1)/(\tilde{y}+y)=\gamma$.
Rearranging this, we obtain the equivalent
\begin{equation}\label{Percy}
\tilde{y}(y-\gamma)=\gamma y+1.
\end{equation}  
Note that $y\not=\gamma$, for otherwise we would get $0=\gamma^2+1$, but $\gamma\not\in\{i,-i\}$.
So we can rearrange \eqref{Percy} to obtain the equivalent
\begin{equation}\label{Egbert}
\tilde{y} = \frac{\gamma y+1}{y-\gamma}.
\end{equation}
Now let $\phi\colon E \to E$ with $\phi(x)=\tilde{x}=x^{p^k}$ so that $\phi^j(x)=x^{p^{j k}}$.
For every nonsingular $2\times 2$ matrix
\[
M=\begin{pmatrix} M_{11} & M_{12} \\ M_{21} & M_{22} \end{pmatrix}
\]
with entries in $E$, let $F_M\colon E\cup\{\infty\} \to E\cup\{\infty\}$ be the linear fractional transformation $F_M(x)=(M_{11} x + M_{12})/(M_{21} x + M_{22})$.
If $M$ and $N$ are both nonsingular $2\times 2$ matrices, it is well known that $F_M \circ F_N=F_{M N}$.
Let $I$ be the $2\times 2$ identity matrix and let
\[
J=\begin{pmatrix} 0 & 1 \\ 1 & 0 \end{pmatrix} \text{ and } K=\begin{pmatrix} 1 & 0 \\ 0 & -1\end{pmatrix}.
\]
Then \eqref{Egbert} says that $\phi(y)=F_{\gamma K + J}(y)$, where we note that $\det(\gamma K+J)=-(\gamma^2+1)\not=0$ (since $\gamma\not\in\{i,-i\}$).
Proceeding by induction, one then proves that for every $j \in \N$, we have
\begin{equation}\label{Bartholomew}
\phi^j(y) = F_{\left(\phi^{j-1}(\gamma) K + J\right) \cdots \left(\phi(\gamma) K + J\right) \left(\gamma K + J\right)}(y),
\end{equation}
where case with $j=1$ has already been shown and if $j>1$, then one can fall back on the earlier cases to compute
\begin{align*}
\phi^j(y) 
& = \phi^{j-1}(F_{\gamma K + J}(y)) \\
& = F_{\phi^{j-1}(\gamma) K + J}(\phi^{j-1}(y)) \\
& = F_{\phi^{j-1}(\gamma) K + J}\left(F_{\left(\phi^{j-2}(\gamma) K + J\right) \cdots \left(\phi(\gamma) K + J\right) (\gamma K + J)}(y)\right),
\end{align*}
and then use the aforementioned identity $F_M\circ F_N=F_{M N}$ with the last formula to obtain \eqref{Bartholomew}.
Since $\phi^n(y)=y^{p^{n k}}$ and $k$ is even and $y \in U_E \subseteq E=\F_{p^{2 n}}$, we know that
\begin{equation}\label{Basil}
y=\phi^n(y)=F_L(y),
\end{equation}
where
\begin{equation}\label{Gregory}
L=\left(\phi^{n-1}(\gamma) K + J\right) \cdots \left(\phi(\gamma) K + J\right) \left(\gamma K + J\right).
\end{equation}
\end{step}
\begin{step}\label{Samuel}
{\bf A simpler formula for the nonsingular matrix $L$:}
Each multiplicand $\phi^j(\gamma) K+J$ in \eqref{Gregory} is nonsingular, since it is the image under an automorphism of the nonsingular matrix $\gamma K+J$.
Thus, $L$ is a nonsingular matrix.
Now notice that $J^2=I$, so that for any $j$ we have $\phi^j(\gamma) K + J= J \left(\phi^j(\gamma) J K  + I\right)$, and thus
\[
L=J \left(\phi^{n-1}(\gamma) J K  + I\right) \cdots J \left(\phi(\gamma) J K  + I\right) J \left(\gamma J K  + I\right),
\]
Then notice that $J (J K) J = J^2 K J = K J = -J K$, so that we can group terms as $J(\phi^j(\gamma) J K + I)J=-\phi^j(\gamma) J K +I$ for odd $j$, and simply keep the terms $\phi^j(\gamma) J K +I$ for even $j$ to obtain
\[
L=J \prod_{j=0}^{n-1} \left((-1)^j \phi^j(\gamma) J K  + I\right),
\]
where the product from $0$ to $n-1$ is of multiplicands that commute with each other.
It is convenient to introduce factors of $-i$ and $i$ to obtain
\[
L=J \prod_{j=0}^{n-1} \left((-1)^j \phi^j(- i \gamma) i J K  + I\right),
\]
where we are using the fact the $\phi$ fixes $-i$ since it is an automorphism iterated an even number of times.
Now let $\sigma_\ell(x_0,\ldots,x_{n-1})$ be the $\ell$th elementary symmetric function of $n$ variables, so that we have
\[
L=J \sum_{\ell=0}^{n} \sigma_{\ell}((-1)^0 \phi^0(-i \gamma),\ldots,(-1)^{n-1} \phi^{n-1}(-i \gamma)) (i J K)^\ell.
\]
Now notice that $(i J K)^2 = - J K J K = J J K K = I$, so that if we define
\[
a_j= \sums{0 \leq \ell < n \\ \ell \equiv j \!\!\!\! \pmod{2}} \sigma_{\ell}((-1)^0 \phi^0(-i \gamma),\ldots,(-1)^{n-1} \phi^{n-1}(-i \gamma)).
\]
for $j \in \{0,1\}$, then we have
\begin{equation}\label{Ernest}
L=J(a_0 I + a_1 i J K)=a_0 J + i a_1 K.
\end{equation}
\end{step}
\begin{step}\label{Veronica}
{\bf Elements of the fiber of $g$ over some $\gamma\in\Gamma\smallsetminus\{i,-i\}$ in terms of $a_0$ and $a_1$:}
Using \eqref{Ernest} in \eqref{Basil}, we obtain
\[
y = \frac{i a_1 y + a_0}{a_0 y - i a_1},
\]
which rearranges to give $a_0 y^2 - 2 i a_1 y - a_0 = 0$.
We cannot have $a_0=0$, since that would force $a_1=0$ (because $y \in U_E$ cannot be zero), and that would make the nonsingular matrix $L$ from \cref{Samuel} vanish.
Therefore, $y$ satisfies the quadratic equation
\[
x^2 - \frac{2 i a_1}{a_0} x -1,
\]
the product of whose roots equals $-1$.
We have shown that the fiber $g^{-1}(\{\gamma\})$ is a subset of the root set of this polynomial, and since we know that $g^{-1}(\{\gamma\})$ is nonempty (by \cref{Susan}) and closed under the map $x \mapsto -1/x$ (by \eqref{Priscilla}), the fiber must be the entire root set of the polynomial.
Thus,
\[
g^{-1}(\{\gamma\})= \left\{i\frac{a_1}{a_0} + \sqrt{1-\left(\frac{a_1}{a_0}\right)^2}, i \frac{a_1}{a_0}-\sqrt{1-\left(\frac{a_1}{a_0}\right)^2}\right\}.
\]
\end{step}
\begin{step}\label{Wilfred}
{\bf $a_0$, $a_1$, and $a_1/a_0$ written in terms of $r_1$, $s_2$, and $r_3$:}
Now we note that
\begin{align*}
r_1(-i \gamma) & = \sum_{\ell=0}^{n-1} \sigma_\ell((-1)^0 \phi^0(-i \gamma),\ldots,(-1)^{n-1} \phi^{n-1}(-i \gamma)) \\
r_2(-i \gamma) & = \sum_{\ell=0}^{n-1} (-1)^{n-\ell} \sigma_\ell((-1)^0 \phi^0(-i \gamma),\ldots,(-1)^{n-1} \phi^{n-1}(-i \gamma)),
\end{align*}
and $n$ is odd, so that $r_1(-i \gamma)+r_2(-i \gamma)=2 a_1$ and $r_1(-i \gamma)-r_2(-i \gamma)=2 a_0$.
Since we proved that $a_0\not=0$ in \cref{Veronica}, we know that $r_1(-i \gamma)\not=r_2(-i \gamma)$ when $\gamma\in\Gamma\smallsetminus\{i,-i\}$.
So then $r_3(-i \gamma)$ is defined and equals $a_1/a_0$.
\end{step}
\begin{step}\label{Yuri}
{\bf Elements of the fiber of $g$ in terms of $\gamma$ and $u$:}
Combining Steps \ref{Veronica} and \ref{Wilfred}, we obtain
\[
g^{-1}(\{\gamma\}) = \left\{i r_3(-i \gamma) + \sqrt{1-r_3(-i \gamma)^2}, i r_3(-i \gamma)-\sqrt{1-r_3(-i \gamma)^2}\right\}
\]
when $\gamma\in\Gamma\smallsetminus\{i,-i\}$.
Along with \cref{Abigail}, this shows that the rule for the function $u$ in the statement of this proposition does indeed map $\Gamma$ into $U_E$ (the domain of $g$), and one can check that the formula for $u(\gamma)$ and the formula claimed for $-1/u(\gamma)$ always have a product of $-1$, thus validating the latter formula.
So our results here and in \cref{Abigail} show us that $g^{-1}(\{\gamma\})=\left\{u(\gamma),-1/u(\gamma)\right\}$, and since fibers of $g$ must cover the domain of $g$, we see that $u$ must be surjective.
\end{step}
\begin{step}
{\bf Fibers of $h$:}
Now suppose that $\beta \in B$.
Since $h(x)=g(x)^2$ for all $x \in U_E$, we have $h^{-1}(\{\beta\})=g^{-1}(\{\sqrt{\beta}\}) \cup g^{-1}(\{-\sqrt{\beta}\})$.
Steps \ref{Abigail} and \ref{Yuri} show that $g^{-1}(\{\sqrt{\beta}\})=\{u(\sqrt{\beta}),-1/u(\sqrt{\beta})\}$, and \eqref{Priscilla} shows that the elements of $g^{-1}(\{-\sqrt{\beta}\})$ are the opposites of the elements of $g^{-1}(\{\sqrt{\beta}\})$.\qedhere
\end{step}
\end{proof}
The fiber of $f_4\circ\dc$ over $0$ has special behavior, and so it is expedient to determine it directly, as we do in the following result.
\begin{lemma}\label{Madeleine}
Let $F$ be a finite field of odd characteristic $p$ and order $p^n$, let $k$ be a nonnegative integer, let $d_1=(p^n+1)/2$, $d_2=(p^k+1)/2$, and suppose that $\gcd(d_2,\cardFu)=1$.
Let $f_4$ be as defined in \cref{Dorothy}.
Let $E$ be the quadratic extension of $F$, let $i$ and $-i$ be the primitive fourth roots of unity in $E$, let $U_E$ be the unit circle of $E$, and let $\dc\colon\du \to F$ be the fiber doubling map.
The fiber $(f_4\circ\dc)^{-1}(\{0\})$ is empty if $q \equiv 1 \pmod{4}$ and equals $\{(i,U_E),(-i,U_E)\}$ if $q \equiv 3 \pmod{4}$.
\end{lemma}
\begin{proof}
By \cref{Irene}\ref{Julius}, there are no elements of the form $(b,\Fu)$ in $(f_4\circ\dc)^{-1}(\{0\})$.
By \cref{Irene}\ref{Julia}, an $x \in U_E$ has $(f_4\circ\dc)(x,U_E)=0$ if and only if $x^{-1}+x=0$, which is true if and only if $x^2+1=0$, whose solutions are the primitive fourth roots of unity, which reside in the cyclic group $U_E$ of order $q+1$ if and only if $4 \mid q+1$.
\end{proof}

\subsection{When $d_1=(q+1)/2$ and $d_2=(3^k+1)/2$ in $x\mapsto x^{d_2/d_1}$ over a field of characteristic $3$ and order $q$}\label{Juliette}
In this section, we further specialize the field which is the domain of our power function to be of characteristic $3$.
The two results here also assume that its order is $q=3^n$ for an odd $n$ and that our power function's fractional exponent is $(q+1)/(3^k+1)$ where $k$ is even and $\gcd(k,n)=1$.
Our first result, \cref{William}, uses symmetry to help us determine the contents of a fiber from partial knowledge of it, and our second result, \cref{Edward}, is a complete determination of the fibers of $f_4\circ\dc$ and $f_4$ with polynomial formulae for the elements in each fiber.
The latter result is used in \cref{Prospero} to determine the fibers of the discrete derivative $\Delta f$.
\begin{lemma}\label{William}
Let $F$ be a finite field of order $q=3^n$ with $n$ odd, let $k$ be an even nonnegative integer with $\gcd(k,n)=1$, let $d_1=(3^n+1)/2$, and let $d_2=(3^k+1)/2$.
Let $f_4$ be as defined in \cref{Dorothy}.
Let $E$ be the quadratic extension of $F$, let $U_E$ be the unit circle of $E$, and let $\dc\colon\du \to F$ be the fiber doubling map.
If $(a,S) \in \du$, then the fiber of $f_4\circ\dc$ that contains $(a,S)$ also contains $(1/a,S)$ but contains no other element of the form $(b,S)$ with $b\not\in \{a,1/a\}$.
\end{lemma}
\begin{proof}
The conditions on $n$ and $k$ make $q\equiv 3\pmod{4}$ and $\gcd(d_2,q-1)=1$ (see \cref{Clementine}\ref{Arthur}), so that we are able to apply \cref{Quentin} and \cref{Irene} in our proof.
Suppose that $S \in \{\Fu,U_E\}$ and that $a,b \in S$ with $(b,S)$ is in the same fiber of $f_4\circ\dc$ as $(a,S)$.

First consider the case where $S=\Fu$, so that $a,b \in \Fu$.
Then $(f_4\circ\dc(b,\Fu))^2=((f_4\circ\dc)(a,\Fu))^2$, so that by \cref{Irene}\ref{Julius} we have $(b^{e_2}+b^{-e_2})^2 = (a^{e_2}+a^{-e_2})^2$, where $e_2=(3^k-1)/2$.
Thus, $b^{e_2}+b^{-e_2} = s (a^{e_2}+a^{-e_2})=(s a^{e_2})+(s a^{e_2})^{-1}$ for some $s \in \{1,-1\}$.
If we had $s=-1$, then by \cref{Simon}, we have $b^{e_2} \in \{-a^{e_2},-a^{-e_2}\}$, i.e., either $(b/a)^{e_2}=-1$ or else $(b a)^{e_2}=-1$; since $e_2$ is even by \cref{Francis}, this would make $-1$ a quadratic residue in $\Fu$, contradicting \cref{Quentin}.
Thus, we must have $s=1$, and so \cref{Simon} shows that $b^{e_2} \in \{a^{e_2},a^{-e_2}\}$.
So \cref{Yolanda} implies that $b \in \{a,-a,1/a,-1/a\}$.

Now consider the case where $S=U_E$, so that $a,b \in U_E$.
Then $(f_4\circ\dc(a,U_E))^2=(f_4\circ\dc(b,U_E))^2$, so that by \cref{Irene}\ref{Julia}, we have
\[
\left(\frac{b^{3^k+1}-1}{b^{3^k}+b}\right)^2 =\left(\frac{a^{3^k+1}-1}{a^{3^k}+a}\right)^2,
\]
so that $b \in \{a,-1/a,-a,1/a\}$ by \cref{Uriah}.

So in every case we have $b \in \{a,-a,1/a,-1/a\}$.
If $b \in \{-a,-1/a\}$, then \cref{Thomas} shows that $(f_4\circ\dc)(b,S)=-(f_4\circ\dc)(a,S)$, and since we assumed that $a$ and $b$ are in the same fiber of $f_4\circ\dc$, this makes $(f_4\circ\dc)(b,S)=(f_4\circ\dc)(a,S)=0$.
If $S=\Fu$, then this is impossible by \cref{Irene}\ref{Julius}, and if $S=U_E$, then \cref{Madeleine} forces $a,b \in \{i,-i\}$, where $i$ is a primitive fourth root of unity, so then $\{-a,-1/a\}=\{i,-i\}=\{a,1/a\}$ and so $b \in \{a,1/a\}$.
So we see that $b \in \{a,1/a\}$ always.
Thus, the fiber of $f_4\circ\dc$ containing $(a,S)$ cannot contain any element of the form $(b,S)$ other than $(a,S)$ and $(1/a,S)$.
And in fact, this fiber must contain $(1/a,S)$ because \cref{Thomas} shows that $f_4\circ\dc(1/a,S)=f_4\circ\dc(a,S)$.
\end{proof}
Now we make complete determination of the fibers of $f_4\circ\dc$ and of $f_4$.
\begin{proposition}\label{Edward}
Let $F$ be a finite field of order $q=3^n$ with $n$ odd, let $k$ be an even nonnegative integer with $\gcd(k,n)=1$, let $d_1=(3^n+1)/2$, let $d_2=(3^k+1)/2$, let $e_2=(3^k-1)/2$, and let $\epsilon$ be a positive integer that is a multiplicative inverse of $e_2/2$ modulo $\cardFu/2$ (which exists by \cref{Yolanda}).
Let $f_4$ be as defined in \cref{Dorothy}.
Let $E$ be the quadratic extension of $F$, let $i$ and $-i$ be the fourth roots of unity in $E$, let $U_E$ be the unit circle of $E$, and let $\dc\colon\du \to F$ be the fiber doubling map.
Let $\eta$ be the extended quadratic character of $F$.
Let $s_1(x),\ldots,s_8(x) \in \Fthree[x]$ and $s_9 \in \F_9[x]$ be
\begin{align*}
s_1(x) & = (1-x^2)^{\frac{q+1}{4}}, \\
s_2(x) & = x^{\frac{q-1}{2}+\frac{(q-3)\epsilon}{4}} \left(s_1(x)+1 \right)^{\frac{(3 q-1)\epsilon}{4}}, \\
s_3(x) & = s_2(x)+s_2(x)^{q-2}, \\
s_4(x) & = \prod_{j=0}^{n-1} \left((-1)^j s_1(x)^{3^{j k}}+1\right), \\
s_5(x) & = \left(s_4(x)-s_4(x)^{q-2}\right)\left(s_4(x)+s_4(x)^{q-2}\right)^{q-2}, \\
s_6(x) & = s_1(s_5(x)), \\
s_7(x) & = \left((x-1)^{\frac{q+1}{2}}-(x+1)^{\frac{q+1}{2}}\right)^{\frac{3^k+1}{2}} \left((x-1)^{\frac{3^k+1}{2}}-(x+1)^{\frac{3^k+1}{2}}\right)^{\frac{q-3}{2}}, \\
s_8(x) & = x^{q-2} s_7(s_6(x)) s_6(x)\text{, and} \\
s_9(x) & = -s_8(x)+i s_5(x).
\end{align*}
We have the following determination of the fibers of $f_4\circ\dc$.
\begin{enumerate}[label=(\roman*)]
\item\label{Alan} $(f_4\circ\dc)^{-1}(\{0\})=\{(i,U_E),(-i,U_E)\}$,
\item\label{Boris} $(f_4\circ\dc)^{-1}(\{1\})=\{(1,\Fu),(1,U_E)\}$, and
\item\label{Christine} $(f_4\circ\dc)^{-1}(\{-1\})=\{(-1,\Fu),(-1,U_E)\}$.
\item\label{Deidre} If $c \in F\smallsetminus\Fthree$ with $\eta(1-c^2)=-1$, then $(f_4\circ\dc)^{-1}(\{c\})=\emptyset$.
\item\label{Eileen} If $c \in F\smallsetminus\Fthree$ with $\eta(1-c^2)=1$, then $s_2(c) \in \Fu\smallsetminus\{1,-1\}$ and $s_9(c) \in U_E\smallsetminus\{1,-1\}$ and
\[
(f_4\circ\dc)^{-1}(\{c\})=\{(s_2(c),\Fu),(1/s_2(c),\Fu),(s_9(c),U_E),(1/s_9(c),U_E)\},
\]
which has precisely four distinct elements.
\end{enumerate}
And we have the following determination of the fibers of $f_4$.
\begin{enumerate}[label=(\roman*),start=6]
\item\label{Fergus} If $c \in F$ with $\eta(1-c^2)=-1$, then $f_4^{-1}(\{c\})=\emptyset$.
\item\label{Geoffrey} If $c \in F$ with $\eta(1-c^2)\not=-1$, then 
\[
f_4^{-1}(\{c\})=\{s_3(c),s_8(c)\}.
\]
If $c\in\Fthree$, then $s_3(c)=s_8(c)=-c$, but if $c\not\in\Fthree$, then $s_3(c)$ and $s_8(c)$ are distinct elements of $F\smallsetminus\Fthree$.
\end{enumerate}
\end{proposition}
\begin{proof}
We break the proof into steps.
\begin{gradus}
{\bf Proof of part \ref{Alan}:}
Since $q=3^n$ with $n$ odd, we have $q=3^n \equiv 3 \pmod{4}$, and also $\gcd(d_2,\cardFu)=1$ by \cref{Clementine}\ref{Arthur}.
Thus, part \ref{Alan} follows from \cref{Madeleine}.
\end{gradus}
\begin{gradus}\label{Renatus}
{\bf Proof of part \ref{Boris}:}
For $S \in \{\Fu,U_E\}$, the fact that $F$ is of characteristic $3$ makes it easy to calculate directly $(f_4\circ\dc)(1,S)=f_4(2)=1$, so both $(1,\Fu)$ and $(1,U_E)$ lie in $(f_4\circ\dc)^{-1}(\{1\})$, and then \cref{William} shows that no other element can lie in this fiber, thus proving part \ref{Boris}.
\end{gradus}
\begin{gradus}
{\bf Proof of part \ref{Christine}:}
In view of \cref{Renatus}, \cref{Thomas} (with $p=3$ and $j=n$) shows that $(f_4\circ\dc)(-1,S)=-1$ for $S \in \{\Fu,U_E\}$, so that $(-1,\Fu)$ and $(-1,U_E)$ lie in $(f_4\circ\dc)^{-1}(\{-1\})$, and then \cref{William} shows that no other element can lie in this fiber, thus proving part \ref{Christine}.
\end{gradus}
\begin{gradus}
{\bf Proof of part \ref{Deidre}:}
Now let us suppose that $c\in F\smallsetminus\Fthree$ with $\eta(1-c^2)=-1$.
Suppose for a contradiction that there were some $(a,S) \in (f_4\circ\dc)^{-1}(\{c\})$.
Then $1-c^2=1-\left((f_4\circ\dc)(a,S)\right)^2$.
If $S=\Fu$, then $a \in \Fu$, and then \cref{Irene}\ref{Julius} shows that $1-c^2=((a^{e_2}-a^{-e_2})/(a^{e_2}+a^{-e_2}))^2$, which is a square of an element of $\Fu$, but this contradicts our supposition that $\eta(1-c^2)=-1$.
So we are forced to assume that $S=U_E$, and then $a \in U_E$, so \cref{Irene}\ref{Julia} shows that $1-c^2=-((a^{3^k+1}-1)/(a^{3^k}+a))^2$, which is nonzero since $c\not\in\Fthree$.
Since $a\in U_E$ makes $a^q=a^{-1}$, we have
\[
\left(\frac{a^{3^k+1}-1}{a^{3^k}+a}\right)^q = \frac{a^{-(3^k+1)}-1}{a^{-3^k}+a^{-1}} = -\left(\frac{a^{3^k+1}-1}{a^{3^k}+a}\right),
\]
so $((a^{3^k+1}-1)/(a^{3^k}+a))^2$ is an element of $\Fu$ whose square roots do not lie in $F$.
Thus, since $-1$ is a quadratic nonresidue of $F$ by \cref{Quentin}, this makes $1-c^2$ a quadratic residue of $F$, again contradicting our supposition that $\eta(1-c^2)=-1$.
This finishes the proof of part \ref{Deidre}.
\end{gradus}
\begin{gradus}\label{Jeremy}
{\bf $s_1(c) \in \Fqr\smallsetminus\{1,-1\}$ with $s_1(c)^2=1-c^2$ when $c\in F\smallsetminus\Fthree$ and $\eta(1-c^2)=1$:}
Suppose that $c\in F\smallsetminus\Fthree$ with $\eta(1-c^2)=1$.
\cref{Bernard} tells us that $s_1(c)$ is a square root of $1-c^2$ and is a quadratic residue.
Since $s_1(c)^2=1-c^2$ and $c\not=0$, we know that $s_1(c)$ is not $1$ or $-1$.
\end{gradus}
\begin{gradus}\label{Karl}
{\bf $(f_4\circ\dc)^{-1}(\{c\}) \cap (\Fu\times\{\Fu\}) = \{(s_2(c),\Fu), (1/s_2(c),\Fu)\}$ when $c\in F\smallsetminus\Fthree$ and $\eta(1-c^2)=1$:}
We continue to suppose that $c\in F\smallsetminus\Fthree$ with $\eta(1-c^2)=1$.
Since $c \in \Fu$ and since \cref{Jeremy} shows that $s_1(c)\not=-1$, we see that $s_1(c)+1$ and $s_2(c)$ are in $\Fu$.
Notice that 
\begin{align}
\begin{split}\label{Horatio}
s_2(c)^{\frac{q-1}{2}}
& = c^{\left(\frac{q-1}{2}+\frac{(q-3)\epsilon}{4}\right)\frac{q-1}{2}} \left(s_1(c)+1\right)^{\frac{(3 q-1)\epsilon}{4} \cdot \frac{q-1}{2}} \\
& = c^{\frac{q-1}{2}},
\end{split}
\end{align}
where the last equality reduces exponents modulo $q-1$ using the fact that $(q-1)/2$ is odd but $(q-3)/4$ and $(3 q-1)/4$ are even.
Furthermore, since $e_2$ is even and $(e_2/2)\epsilon \equiv 1 \pmod{(q-1)/2}$, we see that $e_2(q-1)/2 \equiv 0 \pmod{q-1}$ and $e_2 \epsilon \equiv 2 \pmod{q-1}$ and $e_2 \epsilon (3 q-1)/4 \equiv (q+1)/2 \pmod{q-1}$, so that
\begin{align*}
s_2(c)^{e_2}
& = c^{\frac{q-3}{2}} \left(s_1(c)+1 \right)^{\frac{q+1}{2}} \\
& = \left(\frac{s_1(c)+1}{c}\right)^{\frac{q+1}{2}},
\end{align*}
which is a quadratic residue equal to $t (s_1(c)+1)/c$ for some $t \in \{1,-1\}$ by \cref{Lily}.
Thus,
\begin{align*}
s_2(c)^{e_2}+s_2(c)^{-e_2}
& = t \left(\frac{s_1(c)+1}{c} + \frac{c}{s_1(c)+1}\right) \\
& = t \cdot \frac{(s_1(c)+1)^2+c^2}{c(s_1(c)+1)} \\
& = t \cdot \frac{1-c^2+2 s_1(c)+1+c^2}{c(s_1(c)+1)} \\
& = -\frac{t}{c},
\end{align*}
where we recall that our field is of characteristic $3$ in the last equality.
Then since $(q+1)/2$ is even, we see that $(s_2(c)^{e_2}+s_2(c)^{-e_2})^{-(q+1)/2}=c^{(q+1)/2}$.
Therefore, \cref{Irene}\ref{Julius} (along with the fact that $F$ is of characteristic $3$ with $q=\cardF\equiv 3\pmod{8}$) shows that
\begin{align*}
(f_4\circ\dc)(s_2(c),\Fu)
& = s_2(c)^{\frac{q-1}{2}} c^{\frac{q+1}{2}} \\
& = c^q && \text{by \eqref{Horatio}} \\
& = c && \text{since $c \in F$,} 
\end{align*}
and so $(s_2(c),\Fu) \in (f_4\circ\dc)^{-1}(\{c\})$.
\cref{William} shows that $(1/s_2(c),\Fu)$ is also in this fiber, but no other element of the form $(b,\Fu)$ is.
\end{gradus}
\begin{gradus}\label{John}
{\bf $(f_4\circ\dc)(z,U_E) \in \{c,-c\}$ for $z=u(i s_1(c))\in U_E$ with $u$ from \cref{Uriah} when $c\in F\smallsetminus\Fthree$ and $\eta(1-c^2)=1$:}
We continue to suppose that $c\in F\smallsetminus\Fthree$ with $\eta(1-c^2)=1$.
By \cref{Jeremy}, $1-c^2 \in \Fqr$, so by \cref{Alasdair}, $\eta(c^2-1)=-1$.
Clearly $\eta((c^2-1)+1)=\eta(c^2)=1$ since $c$ is a nonzero element of $F$.
We let $B$ and $\Gamma$ be as defined in \cref{Uriah}, so then $c^2-1 \in B$.
Since $i s_1(c)$ is a square root of $c^2-1$ by \cref{Jeremy}, we have $i s_1(c) \in \Gamma$.
We define the functions $h$ and $u$ as in \cref{Uriah}, and let $z=u(i s_1(c))$, which is in $U_E$ since it is an output of $u$.
Then apply \cref{Uriah} to see that $h^{-1}(\{c^2-1\})=\{z,-z,1/z,-1/z\}$, so that $1+h(z)=c^2$.
Then \cref{Irene}\ref{Julia} shows that $1+h(x)=((f\circ\dc)(x,U_E))^2$ for every $x \in U_E$, so we must have $(f_4\circ\dc)(z,U_E) \in \{c,-c\}$.
\end{gradus}
\begin{gradus}\label{Hugo}
{\bf $r_3(-i(i s_1(c))=s_5(c)$ with $r_3$ from \cref{Uriah} when $c\in F\smallsetminus\Fthree$ and $\eta(1-c^2)=1$:}
We continue to suppose that $c\in F\smallsetminus\Fthree$ with $\eta(1-c^2)=1$.
Recall the definitions of $r_1$, $r_2$, and $r_3$ from \cref{Uriah}, and note that the formula for $z=u(i s_1(c))$ involves the term $r_3(-i (i s_1(c))=r_3(s_1(c))=(r_1(s_1(c))+r_2(s_1(c)))/(r_1(s_1(c))-r_2(s_1(c)))$, and \cref{Uriah} ensures that $r_1(s_1(c))\not=r_2(s_1(c))$ because $i s_1(c) \not \in \{i,-i\}$ by \cref{Jeremy}.
We then define
\[
P(x) = \prod_{j=0}^{n-1} \left((-1)^j s_1(x)^{3^{j k}}-1\right),
\]
and see that $r_1(s_1(c))=s_4(c)$ and $r_2(s_1(c))=P(c)$ with $s_4(c)\not=P(c)$.
If we let $N_{F/\Fthree}$ be the absolute norm from $F$ to $\Fthree$, then
\begin{align*}
s_4(c) P(c)
& = \prod_{j=0}^{n-1} \left((s_1(c)^2)^{3^{j k}}-1\right) \\
& = \prod_{j=0}^{n-1}(-c^2)^{3^{j k}} && \text{by \cref{Jeremy}} \\
& = -N_{F/\Fthree}(c)^2 && \text{$n$ is odd and $\gcd(k,n)=1$} \\
& = -1 && \text{$N_{F/\Fthree}$ maps into $\Fthree^*=\{1,-1\}$}.
\end{align*}
Therefore $s_4(c)$, $P(c) \in \Fu$ with $P(c)=-s_4(c)^{-1}=-s_4(c)^{q-2}$, and we have already observed that $s_4(c)\not=P(c)$, so that $s_4(c)+s_4(c)^{q-2} \in \Fu$.
Then
\begin{align*}
r_3(-i (i s_1(c))
& = \frac{r_1(s_1(c))+r_2(s_1(c))}{r_1(s_1(c))-r_2(s_1(c))} && \text{(see Prop.~\ref{Uriah})} \\
& = \frac{s_4(c)+P(c)}{s_4(c)-P(c)} \\
& =\left(s_4(c)-s_4(c)^{q-2}\right)\left(s_4(c)+s_4(c)^{q-2}\right)^{q-2} \\
& = s_5(c).
\end{align*}
\end{gradus}
\begin{gradus}
{\bf $z=i s_5(c)+ s_6(c)\in U_E$ and $1/z=-i s_5(c)+s_6(c)\in U_E$ when $c\in F\smallsetminus\Fthree$ with $\eta(1-c^2)=1$:}
We continue to suppose that $c\in F\smallsetminus\Fthree$ with $\eta(1-c^2)=1$.
Recall from the previous step that $s_4(c)+s_4(c)^{q-2} \in \Fu$.
By the definition of $s_5$, we have
\begin{align*}
1-s_5(c)^2
& = \frac{\left(s_4(c)+s_4(c)^{q-2}\right)^2-\left(s_4(c)-s_4(c)^{q-2}\right)^2}{\left(s_4(c)+s_4(c)^{q-2}\right)^2} \\
& = \frac{1}{\left(s_4(c)+s_4(c)^{q-2}\right)^2},
\end{align*}
since we are in characteristic $3$ and $s_4(c) \in \Fu$ by \cref{Hugo}.
Thus, $1-s_5(c)^2 \in\Fqr$.
When selecting the square root function $x\mapsto \sqrt{x}$ in \cref{Uriah}, we may use \cref{Bernard} to insist that $\sqrt{x}=x^{(q+1)/4}$ whenever $x \in F$ with $\eta(x)\not=-1$.
By \cref{Hugo}, the term $\sqrt{1-r_3(-i(i s_1(c)))^2}$ in the formula for $u(i s_1(c))$ in \cref{Uriah} becomes
\begin{align}\begin{split}\label{Gertrude}
\sqrt{1-s_5(c)^2}
&= (1-s_5(c)^2)^{(q+1)/4} \\
&= s_1(s_5(c))\\
&= s_6(c).
\end{split}\end{align}
So in fact, \cref{Uriah} and \cref{Hugo} show that $z=u(i s_1(c))=i s_5(c)+s_6(c)$ and also that $-1/z=i s_5(c)-s_6(c)$, so that $1/z=-i s_5(c)+s_6(c)$.
\end{gradus}
\begin{gradus}
{\bf $f_4(s_8(c))=c$ when $c\in F\smallsetminus\Fthree$ with $\eta(1-c^2)=1$:}
We continue to suppose that $c\in F\smallsetminus\Fthree$ with $\eta(1-c^2)=1$.
Using the values of $z$ and $1/z$ from the previous step, we have $\dc(z,U_E)=2 s_6(c)=-s_6(c)$, so that $(f_4\circ\dc)(z,U_E)=f_4(-s_6(c))=-f_4(s_6(c))$ by \cref{Thomas}.
Since $(f_4\circ\dc)(z,U_E) \in \{c,-c\}$ by \cref{John}, this shows that there is some $w \in \{1,-1\}$ such that $f_4(s_6(c))=w c$.
Since we are in characteristic $3$, one can verify that $s_7(a)=f_4(a)$ for every $a \in F$, so $s_7(s_6(c))=w c$.

Since $c \in \Fu$, we see that $s_8(c) \in F$ with
\begin{align}\begin{split}\label{Desmond}
s_8(c)
& = c^{q-2} s_7(s_6(c)) s_6(c) \\
& = c^{q-2} w c s_6(c) \\
& = w s_6(c).
\end{split}\end{align}
Then  
\begin{align}\begin{split}\label{Nicholas}
f_4(s_8(c))
& = f_4(w s_6(c)) \\
& = w f_4(s_6(c)) \\
& = w (w c) \\
& = c,
\end{split}\end{align}
where we use \cref{Thomas} in the second equality.
\end{gradus}
\begin{gradus}\label{Milton}
{\bf $(f_4\circ\dc)^{-1}(\{c\}) \cap (U_E\times\{U_E\}) = \{(s_9(c),U_E), (1/s_9(c),U_E)\}$ when $c\in F\smallsetminus\Fthree$ and $\eta(1-c^2)=1$:}
We continue to suppose that $c\in F\smallsetminus\Fthree$ with $\eta(1-c^2)=1$.
Note that $q \equiv 3 \pmod{4}$ and that $s_5(c), s_8(c) \in F$ since $c \in F$, so that $s_9(c)^q= -s_8(c)^q + i^q s_5(c)^q = -s_8(c)-i s_5(c)$.
And note that
\begin{align*} 
s_9(c) (-s_8(c)-i s_5(c))
&= s_8(c)^2+s_5(c)^2 \\
&= s_6(c)^2+s_5(c)^2 && \text{by \eqref{Desmond}} \\
&= 1 && \text{by \eqref{Gertrude},}
\end{align*}
so $s_9(c)^{-1}=-s_8(c)-i s_5(c)=s_9(c)^q$, showing that $s_9(c) \in U_E$.
Therefore
\begin{align}\begin{split}\label{Solon}
\dc(s_9(c),U_E)
& =(-s_8(c)+i s_5(c))+(-s_8(c)-i s_5(c)) \\
& = s_8(c),
\end{split}\end{align}
since $F$ is of characteristic $3$.
Therefore $(f_4\circ\dc)(s_9(c),U_E) = f_4(s_8(c)) = c$ by \eqref{Nicholas}.

Thus, $(s_9(c),U_E) \in (f_4\circ\dc)^{-1}(\{c\})$ and therefore, by \cref{William}, we also know that $(1/s_9(c),U_E)$ also lies in this fiber and that no other element of the form $(b,U_E)$ lies in this fiber.
\end{gradus}
\begin{gradus}
{\bf Proof of part \ref{Eileen}:}
Together, Steps \ref{Karl} and \ref{Milton} show that if $c\in F\smallsetminus\Fthree$ with $\eta(1-c^2)=1$, then
\[
(f_4\circ\dc)^{-1}(\{c\})=\{(s_2(c),\Fu),(1/s_2(c),\Fu),(s_9(c),U_E),(1/s_9(c),U_E)\}.
\]
Furthermore $s_2(c)\not=1/s_2(c)$ and $s_9(c)\not=1/s_9(c)$ since otherwise $s_2(c)$ or $s_9(c)$ would lie in $\{1,-1\}$, but we already know from parts \ref{Boris} and \ref{Christine} that $(1,\Fu)$, $(-1,\Fu)$, $(1,U_E)$, and $(-1,U_E)$ lie in fibers over $1$ and $-1$, but $c\not\in\{1,-1\}$ here.
Thus, $(f_4\circ\dc)^{-1}(\{c\})$ has precisely four distinct elements, and this completes the proof of part \ref{Eileen}.
\end{gradus}
\begin{gradus}
{\bf Proof of parts \ref{Fergus} and \ref{Geoffrey}:}
Now \cref{Edgar} tells us that for each $c\in F$, we can obtain the fiber of $f_4^{-1}(\{c\})$ by applying $\dc$ to the fiber of $(f_4\circ\dc)^{-1}(\{c\})$.
Thus, part \ref{Fergus} follows immediately from part \ref{Deidre} since any $c \in F$ with $\eta(1-c^2)=-1$ clearly has $c\not\in\Fthree$.
Furthermore, if we apply $\dc$ to the elements of the fibers described in parts \ref{Alan}--\ref{Christine}, we obtain
\begin{align*}
\dc(i,U_E)  &=\dc(-i,U_E) =0 = s_3(0)  = s_8(0)  = -0 \\
\dc(1,\Fu)  &=\dc(1,U_E) =-1 = s_3(1)  = s_8(1)  = -1\\
\dc(-1,\Fu) &=\dc(-1,U_E) =1 = s_3(-1) = s_8(-1) = -(-1),
\end{align*}
where it helps in these calculations to notice that $(q-1)/2$ and $(3^k+1)/2$ are odd while $(q-3)/4$ is even.
So if $c\in\Fthree$, we have $f_4^{-1}(\{c\})=\{s_3(c),s_8(c)\}$ with $s_3(c)=s_8(c)=-c$.
But if $c \in F\smallsetminus\Fthree$ with $\eta(1-c^2)\not=-1$, then $\eta(1-c^2)=1$ and
\begin{align*}
\dc(s_2(c),\Fu) & =\dc(1/s_2(c),\Fu) = s_2(c)+s_2(c)^{-1} = s_3(c) && \text{since $s_2(c) \in \Fu$} \\
\dc(s_9(c),U_E) & =\dc(1/s_9(c),U_E) = s_8(c) && \text{by \eqref{Solon}},
\end{align*}
thus making $f_4^{-1}(\{c\})=\{s_3(c),s_8(c)\}$, where $s_3(c)$ and $s_8(c)$ must be not be in $\Fthree$ (as those elements are already accounted for in other fibers) and $s_3(c)$ and $s_8(c)$ must be distinct because $|f_4^{-1}(\{c\})|=|(f_4\circ\dc)^{-1}(\{c\})|/2=4/2=2$ by \cref{Edgar} and part \ref{Eileen}.\qedhere
\end{gradus}
\end{proof}

\section{Proof of \cref{Victoria}}\label{Prospero}

We now prove our main result (\cref{Victoria}) about the fibers of the discrete derivative $\Delta f$ of our power mapping $f$.
Let $d$ denote our exponent $(3^n+1)/(3^k+1)$ of our power function $x\mapsto x^d$ over $F$ and set $d_1=(3^n+1)/2$ and $d_2=(3^k+1)/2$.
Then $\gcd(d_2,3^n-1)=1$ by \cref{Clementine}\ref{Arthur}, so that $d$ can be regarded as $d_1/d_2$.
Then we may apply Lemmas \ref{Aaron}--\ref{Dorothy} to see that $\Delta f=\sigma^{-1}\circ f_4\circ\tau^{-1}\circ\pi^{-1}\circ\sigma^{-1}\circ\pi$, where $\pi$, $\sigma$, and $\tau$ are maps from $F$ to $F$ with $\pi(x)=1+x^{q-2}$, $\sigma(x)=x^{(3^k+1)/2}$, and $\tau(x)=(x-2)/4=x+1$ (since we are in characteristic $3$).
We should note that $\sigma$, $\pi$, and $\tau$ are all permutations of $F$, and we let $\psi=\pi^{-1}\circ\sigma\circ\pi\circ\tau$, which is a permutation that makes $\Delta f=\sigma^{-1} \circ f_4 \circ \psi^{-1}$.

Lemmas \ref{Fedor} and \ref{Leonid} show that for each $c \in F$, we have
\begin{equation}\label{Nelson}
(\Delta f)^{-1}(\{c\}) = \psi\left(f_4^{-1}(\{\sigma(c)\})\right).
\end{equation}
By \cref{Edward}, we know that our fiber is empty if and only if $\eta(1-\sigma(c)^2)=-1$, i.e., if and only if $\eta(1-c^{3^k+1})=-1$.

For this paragraph, we assume that $c \in F$ with $\eta(1-c^{3^k+1})\not=-1$.
We adopt the definitions of the polynomials $s_1,\ldots,s_9$ from \cref{Edward} and apply that proposition to \eqref{Nelson} to obtain
\[
(\Delta f)^{-1}(\{c\}) = \{\psi(s_3(\sigma(c))), \psi(s_8(\sigma(c)))\}.
\]
It is straightforward to verify that for each $j \in \{1,2,3,4,5,6,8\}$ and $a \in F$, we have $s_j(\sigma(a))=s_j(a^{(3^k+1)/2})=p_j(a)$ as long as one keeps in mind that $(3^k+1)/2$ is odd (because of \cref{Francis}) and that exponents on powers can be reduced modulo $q-1$ as long as positivity is maintained (since we are doing arithmetic in $F$ at this point).
It is not hard to work out that $p_9(a)=\psi(a)$ for all $a \in F$.
From this we conclude that $\psi(s_j(\sigma(c))=p_9(p_j(c))$ for $j \in \{3,8\}$, so that
\[
(\Delta f)^{-1}(\{c\}) = \{p_{10}(c),p_{11}(c)\}.
\]
\cref{Edward} says that $s_3(c)$ and $s_8(c)$ are distinct if and only if $c\not\in\Fthree$.
Since $\sigma$ and $\psi$ are permutations of $F$ and $\sigma$ restricts to a permutation of $\Fthree$, this is the same as saying that $p_{10}(c)$ and $p_{11}(c)$ are distinct if and only if $c\not\in\Fthree$.
For each $a\in\Fthree$, \cref{Edward} says that $s_3(a)=s_8(a)=-a$, and one can easily work out (using the fact that $(3^k+1)/2$ is odd) that $\sigma(a)=a$ and $\psi(a)=a+1$.
Therefore, when $c\in\Fthree$, we have $p_{10}(c)=p_{11}(c)=1-c$.

Our determination of the fibers of $\Delta f$ in the last two paragraphs shows that the formula $\card{(\Delta f)^{-1}(\{c\})}=1+\eta(1-c^2)$ holds true for all $c\in\Fu$, but $\card{(\Delta f)^{-1}(\{0\})}=1$.
Since there are three fibers of size $1$, the remaining $q-3$ fibers must be equally split between sizes $0$ and $2$, since the average fiber size is $1$.
This yields the claimed reduced differential spectrum.
When $n=1$, we have $q=3$, and so all fibers of $\Delta f$ are of size $1$, so that $f$ is PN.
When $n>1$, there are fibers of size $2$, so $f$ is APN.
\hfill\qedsymbol

\appendix
\section{Relation of our family of exponents to a known family}

Consider the family of power functions discussed in \cref{Victoria}.
These are given by $x \mapsto x^{(3^n+1)/(3^k+1)}$ over $\F_{3^n}$ where $n$ is odd, $k$ is even, and $\gcd(k,n)=1$.
(Note that $\gcd((3^k+1)/2,3^n-1)=1$ by \cref{Clementine}, so that our exponent can be interpreted as any positive integer $d$ such that $d \equiv ((3^n+1)/2)((3^k+1)/2)^{-1} \pmod{3^n-1}$, following the convention for fractional exponents set down in the Introduction.)
This section is dedicated to proving that the exponents in this family of power functions are, up to the equivalence defined in the Introduction, the same as those in a theorem of Zha and Wang \cite[Theorem 4.1]{Zha-Wang}, which we now quote.
\begin{theorem}[Zha--Wang, 2010]\label{Sidney}
Let $f(x)=x^d$ be a function over $\GF(3^n)$ where $d$ is even and $(3^m+1)d-2=k(3^n-1)$, $k$ is odd and $\gcd(m,n)=1$. Then $\Delta_d \leq 2$. Furthermore, $f(x)$ is an APN function when $2m < n$.
\end{theorem}
We note that Zha and Wang use $\Delta_d$ to mean the differential uniformity of the function $f(x)$, so $\Delta_d \leq 2$ means that $f$ is either PN or APN.
For the terms in \cref{Sidney} to be defined and make sense, we need to understand $n$, $d$, $m$, and $k$ to be integers with $n$ and $d$ positive and $m$ nonnegative (and \cref{Reginald} below points out that Zha and Wang may have intended $m$ to always be strictly positive).
We now prove the equivalence of our exponents and those of Zha and Wang.
\begin{lemma}\label{Zenobia}
If $n$, $m$, $k$, and $d$ are integers with $n$ and $d$ positive, $m$ nonnegative, $d$ even, $k$ odd, $\gcd(m,n)=1$, and $(3^m+1)d-2=k(3^n-1)$, then $n$ is odd and there is some nonnegative even integer $j$ with $\gcd(j,n)=1$ such that $d$ is equivalent to the fractional exponent $d'=(3^n+1)/(3^j+1)$ over $\F_{3^n}$.
On the other hand, if $n$ and $j$ are integers with $n$ positive, $j$ nonnegative, $n$ odd, $j$ even, and $\gcd(j,n)=1$, and if $d'$ is the fractional exponent $(3^n+1)/(3^j+1)$ in a power function over $\F_{3^n}$, then there are integers $m$, $k$, and $d$ with $m$ and $d$ positive, $d$ even, $k$ odd, $\gcd(m,n)=1$, and $(3^m+1)d-2=k(3^n-1)$ such that $d'$ is equivalent to $d$ over $\F_{3^n}$.
\end{lemma}
\begin{proof}
Suppose $n$, $m$, $k$, and $d$ are integers with $n$ and $d$ positive, $m$ nonnegative, $d$ even, $k$ odd, $\gcd(m,n)=1$, and 
\begin{equation}\label{Jane}
(3^m+1)d-2=k(3^n-1).
\end{equation}
Then notice that since both $3^m+1$ and $d$ are even, we have $(3^m+1) d-2 \equiv 2 \pmod{4}$.
Thus, \eqref{Jane} forces $3^n-1 \equiv 2 \pmod{4}$, which forces $n$ to be odd.
Since $k$ is odd, divide \eqref{Jane} by $2$, reduce modulo $3^n-1$, and rearrange to get
\begin{equation}\label{Tabitha}
\frac{3^n+1}{2} \equiv \left(\frac{3^m+1}{2}\right)d  \pmod{3^n-1}.
\end{equation} 

If $m$ is even, set $j=m$ and $d'=d$ and note that $\gcd((3^j+1)/2,3^n-1)=1$ by \cref{Clementine}, so that $(3^j+1)/2$ is invertible modulo $3^n-1$, and then \eqref{Tabitha} shows that $d=d'$ can be regarded as the fractional exponent $(3^n+1)/(3^j+1)$ in a power function over $\F_{3^n}$.
Then recall that $j=m$ is nonnegative, even, and has $\gcd(j,n)=1$, so that we have proved our first claim in the case where $m$ is even.

If $m$ is odd, pick some positive odd integer $u$ such that $u n-m$ is nonnegative, set $j=u n-m$ (which is even since $u$, $n$ and $m$ are odd), and let $d'=3^m d$ (which is even since $d$ is even).
Then $3^j d' = 3^{u n} d \equiv d \pmod{3^n-1}$.
Thus, substituting $3^j d'$ in for $d$ in \eqref{Tabitha} yields
\[
\frac{3^n+1}{2} \equiv (3^{u n}+3^j)  \frac{d'}{2} \equiv \left(\frac{3^j+1}{2}\right) d' \pmod{3^n-1}.
\]
Then $\gcd((3^j+1)/2,3^n-1)=1$ by \cref{Clementine}, so that $(3^j+1)/2$ is invertible modulo $3^n-1$, and so our congruence shows that we can regard $d'$ (which, being $3^m d$, is equivalent to $d$) as the fractional exponent $(3^n+1)/(3^j+1)$ in a power function over $\F_{3^n}$.
Recall that $j$ is nonnegative and even, and note that $\gcd(j,n)=\gcd(u n-m,n)=\gcd(m,n)=1$.
This completes the proof of our first claim.

Now we prove the second claim, so suppose that $n$ and $j$ are integers with $n$ positive, $j$ nonnegative, $n$ odd, $j$ even, $\gcd(j,n)=1$, and suppose that $d'$ is the fractional exponent $(3^n+1)/(3^j+1)$ in a power function over $\F_{3^n}$, which is defined because \cref{Clementine} makes $\gcd((3^j+1)/2,3^n-1)=1$ so that $d'$ can be regarded as a positive integer with $d' \equiv ((3^n+1)/2)((3^j+1)/2)^{-1} \pmod{3^n-1}$.
Note that $d'$ must be even since \cref{Francis} shows that $v_2(3^n+1) > v_2(3^j+1)$.
Set $d=d'$, which is a positive even integer, and set $m=j+n$, which is positive and has $\gcd(m,n)=\gcd(j,n)=1$.
Then
\begin{align*}
\left(\frac{3^m+1}{2}\right) d
& = \left(\frac{3^{j+n}+1}{2}\right) d'\\
& = \left(\frac{3^j+1}{2}\right) d' + \left(\frac{d'}{2}\right) 3^j (3^n-1) \\
& \equiv \left(\frac{3^j+1}{2}\right) d' \pmod{3^n-1} && \text{(because $d'$ is even)} \\
& \equiv \left(\frac{3^n+1}{2}\right) \pmod{3^n-1}.
\end{align*}
So there is some $\ell \in \Z$ such that
\[
\left(\frac{3^m+1}{2}\right) d = \frac{3^n+1}{2} + \ell(3^n-1),
\]
and we set $k=2\ell+1$, which is an odd integer with
\[
(3^m+1) d -2 =k(3^n-1).\qedhere
\]
\end{proof}
\begin{remark}\label{Reginald}
\cref{Zenobia} shows that all exponents considered in our main result, \cref{Victoria}, are equivalent to exponents considered in \cref{Sidney} of Zha and Wang, and vice-versa.
\cref{Sidney} says that the power functions for those exponents are always PN or APN, and adds that they are certainly APN if $2 m < n$.
\cref{Victoria} shows that the power function is always PN when $n=1$, so \cref{Sidney} would be in error if we set $n=1$ and $m=0$, so perhaps its authors intended $m$ to be strictly positive so that their condition $2 m < n$ would force $n \geq 3$, but they did not state this assumption in their theorem.
In fact, \cref{Victoria} also shows that the power functions are always APN when $n > 1$, so the condition $2 m < n$ in \cref{Sidney} is not mathematically necessary for APNness when $n>1$.
\end{remark}

\section*{Acknowledgement}

The authors thank the anonymous reviewers and Alexander Pott, the editor handling the paper, for their helpful suggestions.

\end{document}